\documentclass[12pt]{amsart}
\usepackage{amsmath}
\usepackage{amsthm}
\usepackage{amssymb}
\usepackage{amscd}
\usepackage{bbold}
\usepackage[pdftex]{graphicx}
\usepackage{enumitem}
\usepackage{subfigure}
\usepackage{graphicx}
\usepackage{enumitem}
\usepackage{subfigure}
\usepackage{cite}
\usepackage{lscape} 
\usepackage{indentfirst}
\usepackage{multirow}
\usepackage{tabls}
\usepackage{wrapfig}
\usepackage{slashbox}
\usepackage{longtable}
\usepackage{supertabular} 

\newtheorem{lem}{Lemma}[section]
\newtheorem{thm}[lem]{Theorem}
\newtheorem{prop}[lem]{Proposition}
\newtheorem{cor}[lem]{Corollary} 

\theoremstyle{definition}
\newtheorem{defn}[lem]{Definition}
\newtheorem{example}[lem]{Example}

\newtheorem{rem}[lem]{Remark}

\newcommand{\CC}{{\mathbb C}}

\newcommand{\FF}{{\mathbb F}}

\newcommand{\HH}{{\mathbb H}}

\newcommand{\RR}{{\mathbb R}}

\newcommand{\Fcal}{\mathcal{F}}

\def\benm{\begin{enumerate}}            
\def\eenm{\end{enumerate}}              
\newcommand{\norm}[1]{\left\Vert #1\right\Vert}         
\newcommand{\inner}[2]{\langle {#1,#2}\rangle}    

\title{Reactive sensing and multiplicative frame super-resolution}

\date{\today}

\subjclass[2010]{42C15}

\keywords{Radiative, dominant, and harmonious
sensing scenarios; multiplicative frames; DFT; super-resolution; dimension reduction}

\author{John J. Benedetto}
\address{Norbert Wiener Center\\
         Department of Mathematics \\
         University of Maryland \\
         College Park, MD 20742 \\
         USA}
\email{jjb@umd.edu}
\urladdr{http://www.math.umd.edu/\textasciitilde jjb}

\author{Michael R. Dellomo}
\address{Norbert Wiener Center\\
         Department of Mathematics \\
         University of Maryland \\
         College Park, MD 20742 \\
         USA}
\email{mdellomo@umd.edu}
\begin{document}

\newcounter{bean}





\begin{abstract}
The problem is to evaluate the behavior of an object 
when primary sources of
information about the object become unavailable, so that any information must be obtained from
the intelligent use of available secondary sources. This evaluative process is
{\it reactive sensing}. Reactive sensing is initially viewed in terms of spatial {\it super-resolution}.
The theory of reactive sensing is based on two 
equivalent ideas,
one physical and one mathematical. 
The {\it physical idea} models volume, e.g., engine volume
in the case of analyzing engine health, 
and the sensitivity of sensors to such volume.
The {\it mathematical idea} of multiplicative frames
provides the factorization theory to compare quantitatively
such volume and sensitivity. This equivalence is the
foundation for reactive sensing theory and its implementation.




\end{abstract}

\maketitle



\section{Introduction}
\label{sec:intro}


\subsection{Background and modeling}
\label{sec:back}

Sensing problems such as those
dealing with RADAR, SONAR, and general engine health, in the context of disabled primary
sensors and in noisy environments, led us to the theory of {\it reactive sensing}. Our 
formulation of this theory necessitates the understanding of secondary sensors in evaluating
primary objects. With this point of view, we see that reactive sensing can be thought of in terms
of spatial super-resolution, e.g., \cite{dono1992}, \cite{chau2001}, \cite{ParParKan2003}, 
\cite{PusKne2005}, \cite{gree2009}, \cite{lind2012}. (This is in contrast to
recent advances in spectral super-resolution, e.g., \cite{CanFer2013}, \cite{CanFer2014}.\cite{BenLi-2018}.) 
In fact, the secondary 
sensors can be 
considered analogous to
the role of obtaining a high resolution (HR) image from observed multiple low-resolution (LR) images. 
In this case, the multiple LR images represent different ``snapshots" of the same scene, and can be combined
to give the desired HR image, see \cite{chau2001}, Chapter 4. In our case, the LR images correspond to secondary
sensors and the HR image corresponds to the primary object, that could be disabled or whose primary sensor
is not functioning.

Notwithstanding this significant connection between our formulation of reactive sensing
and spatial super-resolution, we do not use any of the usual super-resolution 
methodology in developing our theory.
For example, in the case of engine health, the
understanding of secondary sensors is based on a {\it physical modeling} of 
inherent engine volume or vibration and the sensitivity of sensors to such volume,
We show that this physical modeling is equivalent to a {\it mathematical modeling}
formulated in terms of what we call multiplicative frames.
This also leads us to the concept of
{\it sensing scenario health space} determined by data dependent dimension
reduction. 



\subsection{Idea and techniques}
\label{sec:idea}
Our main objective 
in reactive sensing is to evaluate the behavior of an object when primary sources of
information about the object become unavailable so that any information must be obtained from
the intelligent use of available secondary sources. 
For example, if our object is an engine evaluated by a sensor and the sensor
is disabled, we wish to quantify to what extent neighboring sensors, that have primary tasks of
their own, can evaluate the behavior of the engine.

The idea and techniques
we shall introduce to understand reactive sensing, at the level of 
obtaining quantifiable and computationally useful results,
involve an interleaving of the following:

\begin{itemize}
\item The theory of frames, see Section \ref{sec:frames} and Subsection \ref{sec:bmfm}; and
\item Sensing scenario health space, see Section \ref{sec:mathmodel}.
\end{itemize}

The {\it theory of frames} has a long history, based on the work of Paley-Wiener, Beurling,
and Henry Landau, and was explicitly developed by Duffin and Schaeffer (1952). In recent decades
there has been an explosion of activity related to the emergence of Gabor (time-frequency) and
wavelet theories, as well as to applicability in the context of noise reduction, robust signal
decomposition, and numerical stability. See \cite{daub1992},
\cite{BenFra1994}, Chapters 3 and 7, \cite{chri2016}, and \cite{KovChe2007a}, 
\cite{KovChe2007b}.
The interaction of frame theory with the
data sets we
analyze 
has led us to develop a {\it theory of multiplicative frames}.


{\it Health space} can be described functionally as a space where decision
making algorithms operate to provide required data analysis. Its
dimension will be given by the number of parameters required to
perform the analysis. 
The name health
space is used because of questions about machine health, but the same notion can
be applied in other detection and 
classification contexts, see
Definition \ref{defn:sep} and Remark \ref{rem:health}
for specifics on our point of view.





\subsection{Theme and implementation}
\label{sec:theme}

A natural strategy inherent in implementing the idea and technique
of Subsection \ref{sec:idea}
is to choose an efficient method of finding
relevant data and discarding redundant information and noise.
While this
paradigm may work well when all of the data streams are available and all of
the sensors are operating, problems can arise when some
of the data streams are unavailable because of various
possible failures.  The {\it theme} of {\it reactive
sensing} is to construct {\it mappings} of the data streams that are
robust under sensor failure.
A critical aspect of our mathematical
approach
for this construction, necessitating the use 
of multiplicative frames, is that
an individual sensor may be capable of
reporting on parameters which are not its primary 
responsibility. 

Our theme leads to a genuine implementation in Section \ref{sec:datasim} by the 
following process. Once we have shown the equivalence of our physical modeling 
with multiplicative frames
(Sections \ref{sec:mathmodel} and \ref{sec:multframe}), we define
{\it frame mappings} (Definition \ref{defn:BFmappings} b) which allow us to
prove fundamental theorems on the existence of multiplicative frames in Section \ref{sec:rst}.
The definition is technical but motivated by the inherent over-completeness 
of frames. The definition itself and these theorems are essential for the quantitative results
in Section \ref{sec:datasim}, as well as in a host of other applications.



\subsection{Outline}
\label{sec:outline}

We begin 
in Section \ref{sec:ex} with some relevant examples.
These examples provide the backdrop for the mathematical model we 
formulate in Section \ref{sec:mathmodel} for the physical setting
of a reactive sensing scenario.

It is
in Section \ref{sec:mathmodel} that we introduce the notions of a separable sensing
scenario and health space, that are essential for a useful theory of
reactive sensing. Further, the physical notions of radiativity and dominance are quantified
in terms of the notions of volume and sensitivity used in the
definition of the model. The examples in Section
\ref{sec:ex} and the consequent model of Section \ref{sec:mathmodel}
lead to the theory of frames
as a natural tool for effective reactive sensing. 

Subsection \ref{sec:frame} gives a
background on frames, with a comparison to bases, in the context of our setting.
It is in Subsection \ref{sec:multframe} that we introduce multiplicative frames. This is a
mathematical concept directly formulated because of the physical implications of the
model of Section \ref{sec:mathmodel}, see especially 
the paragraph before Definition \ref{defn:sep}, Definition \ref{defn:sep} itself,
Example \ref{ex:dft}, Remark \ref{rem:physalpha}, 
Remark \ref{rem:rad}, and Remark \ref{rem:primrespnonop}. In fact, our 
mathematical constructions of multiplicative frames in Theorem \ref{prop:mf1}, 
Corollary \ref{prop:mf2},
and Theorem \ref{thm:frame}, that 
of themselves are
independent of reactive sensing, {\it depend essentially} on the reactive
sensing ideas of radiativity and dominance 
defined in Definition \ref{defn:raddom}.

Section \ref{sec:rst} presents the theory of reactive sensing in which multiplicative frames 
play an essential role. It is here that we define and analyze basis and frame mappings.
In the context of super-resolution, individual LR sensor outputs that are combined
by means of basis mappings do not improve resolution, i.e., do not lead to reconstructing 
the primary HR object, described in Subsection \ref{sec:back}. On the other hand, frame
mappings, by their over-complete nature, piece together LR sensor outputs to optimize the chances 
of such HR reconstruction.

The DFT plays a fundamental role in our exposition and for our explanations.
As such, we 
give DFT examples throughout as various notions are introduced.
In particular, we give a simplified version of the turbine assembly example,
Example \ref{ex:turbine}, in terms of the DFT.

In Subsection \ref{sec:data} we construct a data base, and in 
Subsection \ref{sec:sim}, we give DFT turbine simulations 
that give proof of concept of our reactive sensing theory 
using this data base. 
Using the signal-to-noise (SNR) as a metric, we
verify improvement in fault detection in the case of sensor failure. We also note 
that this improvement comes
at the expense of lowering 
SNR in a controlled fashion when all sensors are working.
This is the content of Subsection \ref{sec:snr}.
 

The epilogue, Section \ref{sec:epilogue}, gives a summary of the salient features of 
reactive sensing with remarks on future tasks. 


\section{Examples}
\label{sec:ex}

\begin{example}[RADAR]
\label{ex:radar} Consider a collection of RADAR sites which is responsible for
detecting incoming targets in a given sensing scenario, see Figure~\ref{fig:radar}.
Doppler and range
returns for a target might be at the limit of detectability for
RADARs  $A$ and $C$ but be within the main analysis area for RADAR
$B.$ However, if RADAR $B$ is disabled, the data from RADARs $A$ and $C$ might
be combined to give a noisier but adequate return from the target.
Indeed, for the situation described in Figure~\ref{fig:radar}, a target
at $X$ can be detected by both RADARs $A$ and $B.$ While the return from
RADAR $A$ is noisier than that from $B,$ it may still be sufficient to
provide adequate detectability, see Example
\ref{ex:sonradharm} dealing with non-harmonious scenarios.
 In fact, this may allow RADAR $B$ to be
temporarily repurposed, repaired, or relocated as required, while still 
maintaining an acceptable level of coverage, see \cite{skol1980}, 
\cite{leva1988}, \cite{cole1992}, \cite{stim1998}, \cite{LevMoz2004}, \cite{HolRicSch2010}.
\end{example}
\begin{figure}[htbp]
\centering
\includegraphics[height=0.5\paperheight]{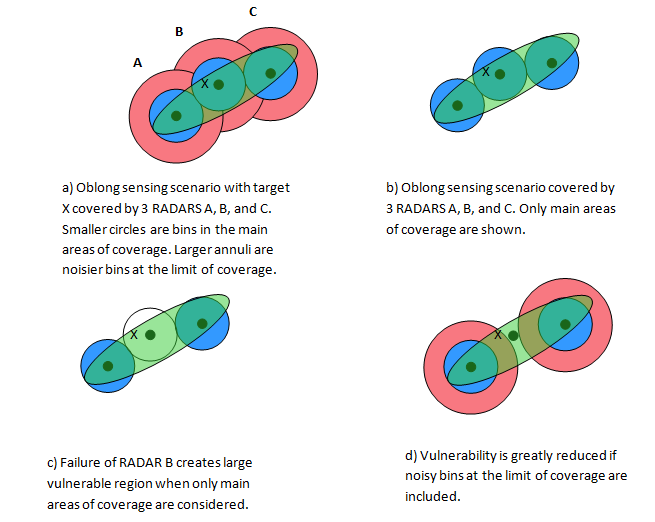}
\caption{RADAR scenario}
\label{fig:radar}
\end{figure}
\begin{example}[SONAR]
\label{ex:sonar} Similarly, consider a SONAR sensing scenario, where a section of
coastal waters is observed by two SONAR arrays, $A$ and $B,$ see Figure~\ref{fig:sonar}.
Each array has
main beams, where the signal strength is high compared to the noise, as
well as beams near endfire, where there is
significantly more noise, see, e.g., \cite{vant2002}.
These beams close to endfire are illustrated in 
Figure~\ref{fig:sonar} by the boundary "petals" in each collection of beams.
If we consider both of these arrays as sensors,
then array $A$ would be primarily responsible for reporting a threat in
its main beams that might be near the endfire beams of array $B$. However, in
the event of a failure in array $A,$ array $B$ could also sense the target,
although with some degradation due to the additional noise, see
Example \ref{ex:sonradharm} dealing with harmonious scenarios.
\end{example}

\begin{figure}[htbp]
\centering
\includegraphics[height=0.2\paperheight]{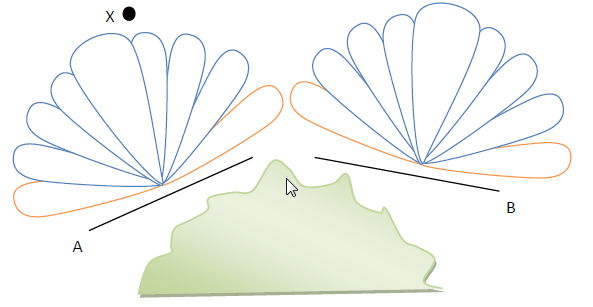}
\caption{Target $X$ as sensed by SONAR arrays $A$ and $B$}
\label{fig:sonar}
\end{figure}
\begin{example}[The DFT and a multi-sensor scenario]
\label{ex:turbine}
A third example concerns the mechanical health of a complex
machine. Consider a machine with several rotating turbines, each
attached to several gears and additional rotating shafts to form a
collection of turbine assemblies, e.g., \cite{dell1992}, \cite{dell1999}. Assume there are
vibration sensors
attached to each turbine assembly and that each assembly has unique
spectral characteristics. A concrete example might be a multi-engine
airplane where each of the engine's turbines is associated
with a given sensor, see Section \ref{sec:datasim}.

The sensor attached to a given turbine assembly will have primary
responsibility for the frequencies associated with that
assembly. This notion of primary responsibility will be 
quantified more precisely in Subsection \ref{sec:sepsensscen},
and takes into account that 
each sensor will also pick up vibrations from other
nearby assemblies at a substantially lower volume. If a sensor fails,
it may be possible to use data from the remaining sensors to report
characteristics of the failed sensor's turbine assembly. 

It is this point of view that led us to 
introduce the theory of frames. In fact, this type
of data
reconstruction is only possible if the fault detection scheme operates
with a frame theoretic representation of the vibration parameters. With a basis
representation, the loss of a sensor causes a loss of all data associated
with that turbine assembly. Indeed, for an airplane, it is important
to distinguish between an impact which causes the loss of a sensor and
an impact which causes the loss of an engine.

The DFT provides a natural technical tool with which to analyze this particular
multi-sensor scenario,
see the DFT Example \ref{ex:dft} and Section \ref{sec:datasim}.
\end{example}


\section{ A mathematical model for reactive sensing}
\label{sec:mathmodel}

\subsection{Separable sensing scenario}
\label{sec:sepsensscen}
We wish to analyze the sensing problems described in 
Subsection \ref{sec:theme} and Section \ref{sec:ex} with a mathematical model. 
To this end, we let $S$ be a set, that we call a {\it set of parameters}. 
To $S$ we associate $N$ {\it sensors},
$s_j,\  j=1,...,N,$ that map subsets of $S$ to their values at
a fixed time $k = 1,...,K$. Each sensor $s_j$ is defined on some
subset $T_j$ of
$S$. To evaluate completely the
impact of the parameters in $S$, we require that
$\bigcup_{j=1}^{N}T_j = S$, i.e., the $T_j, \, j = 1, \ldots, N$, form 
a {\it covering} of $S$.

Each sensor $s_j$ will bear {\it primary
responsibility} for reporting values on some subset $S_j \subseteq
T_j$ in the sense that $S_j \bigcap S_{\ell} = \emptyset,$ where
$1\leq j,\ell \leq N$ and $j \neq \ell,$ and the values of $s_j$ on
$T_j \setminus S_j$ reflect information gathered by $s_j$ for parameters
disjoint from $S_j.$ 
We {\it assume} that the primary responsibility for each parameter is
given to one sensor, and that 
the $S_j, j=1,\ldots N,$ 
form a {\it partition} of $S$, i.e., $\bigcup_{j=1}^{N}S_j = S$ and $S_j \bigcap S_{\ell} = \emptyset,$
where
$0\leq j,\ell \leq N$ and $j \neq \ell.$

The cardinality of any set $X$ is denoted by ${\rm card}\,(X),$ 
and $M$ will denote ${\rm card}\,(S).$ 

In this formulation of a {\it sensor} $s_j$ as a mapping, we have to be precise 
about the role of {\it time} $k.$ As such, each $s_j$
assigns values to the
elements of $T_j$ at a given time index $k$, i.e., $s_j:T_j\times\{1,...,K\}
\rightarrow {\mathbb C}^M$, where we write $s_j(f,k) = v_{j,k}(f).$ Further,
if $f \notin T_j,$ then we define $s_j(f,k) = 0.$
Thus, $v_{j,k} \in {\mathbb C}^M.$ We use the notation, $f \in \{1, \ldots M\},$ 
to think of the 
frequency domain of the DFT, even though our theory is far more general.

Given a set $S$ with partition $\{S_j\}$,
covering $\{T_j\}$, and
mappings $s_j.$  We refer to $S$ as a {\it sensing scenario},
see Figure \ref{fig:sepscen}.

\begin{figure}[htbp]
\centering
\includegraphics[height=0.22\paperheight]{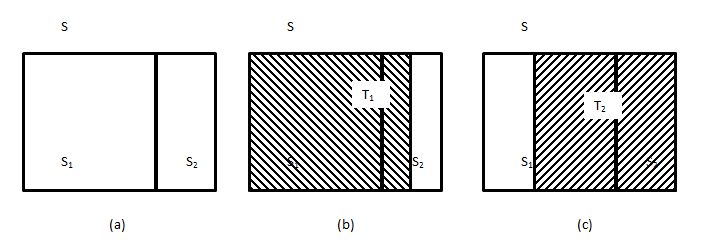}
\caption{Sensing scenario $S:$ In (a), $S$ is the rectangle
with two disjoint subsets, $S_1$ and $S_2;$ in (b), $T_1$ is the
hashed area; and in (c), $T_2$ is the hashed area.}
\label{fig:sepscen}
\end{figure}

If a parameter $f$ is in $T_j$ then, at each time $k,$ sensor $s_j$ assigns a value to it.
We can think of this value as the response that sensor is reporting for the parameter,
$f,$ at time, $k.$ This process is usually the natural result of two effects: the
inherent {\it intensity} (or {\it volume} or {\it loudness}) of the scenario at which the parameter is being generated, 
and the {\it sensitivity} of the sensor
to that parameter. In this regard see Example \ref{ex:dft} and Remark \ref{rem:rad}. As such, we make
the following definition (Definition \ref{defn:sep}). The factorization in this definition and its physical motivation
with regard to volume and sensitivity
are the rationale for our notion of {\it multiplicative frames},
defined in Section \ref{sec:multframe}. The factorization in Definition
\ref{defn:sep}, part {\it b}, ``without
the hats'' is what occurs in {\it a health space}, ${\mathbb C}^n.$ 

\begin{defn}[Separable sensing scenario and health space]
\label{defn:sep}
Let $S$ be a sensing scenario.

 {\it a.} $S$ is {\it pre-separable} if for each $1 \leq f \leq M$, $1 \leq j \leq N$,
and $1 \leq k \leq K$, we have the factorization,
$$
v_{j,k}(f) = {\widehat \gamma}_j(f)\, {\widehat \alpha}_k(f).
$$

{\it b.} A pre-separable 
sensing scenario $S$ is {\it separable} if there are a positive integer $n$ and a mapping,
$$
     H : {\mathbb C}^M \longrightarrow {\mathbb C}^n,
$$
such that,
for each $i = 1,\ldots, n$, $j = 1,\ldots, N$,
and $k = 1,\ldots, K$, the image, $H(v_{j,k}) \in {\mathbb C}^n$, factors
as 
$$
      H(v_{j,k})(i) = {\gamma}_j(i)\, {\alpha}_k(i), 
$$
where ${\alpha}_k, {\gamma}_j \in {\mathbb C}^n$, and such that
$H(0) = 0 \in \CC^n$, see Definition \ref{defn:BFmappings}.

$H$ is a {\it health space mapping}, and ${\mathbb C}^n$ is 
{\it a health space}.
\end{defn}


\begin{rem}[Rationale for health space] 
\label{rem:health}

{\it a. Motivation.}
The motivation behind formulating the notion of health space is to construct
a lower dimensional, space, ${\mathbb C}^n$, in which we can effectively analyze 
a separable sensing scenario $S$.
The dimension $n$ 
is related to the 
general theory of dimension reduction, although
our implementation in Section \ref{sec:datasim} 
only uses a type of PCA.
The health space ${\mathbb C}^n$ and health space mapping
$H$ are part and parcel of the concept of a sensing scenario. 
Indeed, if the number of sensors and parameters is small enough to be analyzed without 
dimension reduction, there would be no need to develop the theory further. Hence, we shall always  
consider a {\it sensing scenario  with health space } 
${\mathbb C}^n$ {\it and health space mapping} $ H : {\mathbb C}^M \longrightarrow {\mathbb C}^n$.

{\it b. The size of $n$.}
The dimension $n$ should be taken as small as possible, and we must have $n < NM$
or there would be no point in applying reactive sensing theory. Usually, the size of
$n$ is determined by the analysis algorithms that will be applied. 

For example,
suppose we have a scenario where the frequency domain output from a group of
vibration sensors is used to diagnose the mechanical health of a device. There may
be many components with complex spectral characteristics that need to be
evaluated to detect or classify a fault. In this case $n$ may be comparatively large
while the number of sensors, $N$, might not be. This is precisely the scenario in
Section 6 where $n = 28$ and $N= 4$. 

On the other hand, consider a collection of inexpensive
sensors each of which is detecting an energy spike. This might be
the case if the sensors were sprinkled across a roadway with the intent
of monitoring the amount of traffic. Here $n$ could be small (maybe even 1), while
$N$ could be large. 

\end{rem}


\subsection{Examples of separable sensing scenaios}
\label{sec:sss}

\begin{example}[A DFT separable sensing scenario]
\label{ex:dft}
{\it a.} A useful example,
of which Example \ref{ex:turbine} is a prototype,
results from considering sensors $s_j$ that
perform a DFT on blocks of $2^{r}$ points. At a time, $k,$ for $k=1,\ldots,
K,$ each sensor, $s_j$ , will report $2^{r}$ values for the DFT
frequency bins. 
Thus, we can define $S$ to be the set of $2^r$ parameters, $f,$
and so $M = {\rm card}(S) = 2^r.$ The sets, 
$T_j$ and $S_j,$ will be a covering and partition of $S,$ respectively.
For example, if we have acoustic or vibration sensors, it may be that
$T_j = S$ since all sensors hear every frequency, but the sets $S_j$
consist of only those frequencies for which sensor $s_j$ is the best value,
e.g., the loudest or highest SNR value. It is for this reason that we 
introduced the notion of {\it primary responsibility}.

{\it b.} In the case of
sensors attached to the turbine assemblies of Example \ref{ex:turbine}, we shall provide
numerical simulations in Section \ref{sec:datasim} quantifying our formulation of reactive sensing.
In fact, we consider 
each $S_j$ in the sensing scenario $S$ as the set of frequencies generated by
the $j$-th engine. 
Thus, each $s_j$ is associated with the same engine and maps the frequencies to their 
values recorded by the sensor. The factor ${\widehat \gamma}_{j}(f) \in {\mathbb C}$
represents the ability of $s_j$ to "hear" the frequency $f$ of $S_j;$ and the factor 
${\widehat \alpha}_{k}(f) \in {\mathbb C}$
represents the "loudness" of the scenario $S$ at frequency $f$ and time $k,$ 
see Remark \ref{rem:rad} for a fuller treatment of this mathematical and physical modeling.
\end{example}



\begin{prop}[Linear separable sensing scenario -- constant ${\widehat \alpha}_k (f)$]
\label{prop:sep1}
Let $S$ be a pre-separable sensing scenario with the property that,
for each $1 \leq k \leq K$,
\begin{equation}
\label{eq:constcond}
    \exists\,{\widehat \alpha}_k \in {\mathbb C} \; {\rm such \, that} \; \forall f,f' \in \{1,...,M\}, 
    \quad {\widehat \alpha}_k (f) = 
    {\widehat \alpha}_k (f') = {\widehat \alpha}_k. 
\end{equation}
Let $H$ be a linear mapping defined by
\[
        H: {\mathbb C}^M \rightarrow {\mathbb C}^n,
\]
\[
          H(x) = Ax, 
\]
where A is an $n\times M$ matrix of complex numbers. 
Then, $S$ is a separable sensing scenario and $H$ is a health space mapping in the sense
of Definition \ref{defn:sep}.
\end{prop}
\begin{proof}
Let $A = [ a_{i,f}]$ where $i \in \{1,...,n\}$ and $f \in \{1,...,M\}$. Since $S$ is pre-separable 
and by (\ref{eq:constcond}), we have
$v_{j,k}(f) = {\widehat \gamma}_j(f){\widehat \alpha}_k(f) = {\widehat \gamma}_j(f){\widehat \alpha}_k.$
Therefore, for each $j = 1, \ldots, N$ and $k = 1, \ldots, K$, we can apply $H$ as follows:
\[ 
H(v_{j,k})(i) = \sum_{f=1}^M a_{i,f} v_{j,k}(f) =  
\sum_{f=1}^M a_{i,f}{\widehat \gamma}_j(f){\widehat \alpha}_k 
=  \big( \sum_{f=1}^M a_{i,f}{\widehat \gamma}_j(f) \big) {\widehat \alpha}_k 
= \gamma_j(i) \alpha_k(i),
\]
where we set $\gamma_j = A{\widehat \gamma}_j \in {\mathbb C}^n$ and 
$\alpha_k(i) = {\widehat \alpha}_k.$
This completes the proof.
\end{proof}

\begin{rem}[Physical systems with and without constant ${\widehat \alpha}_k (f)$]
\label{rem:physalpha}
Condition \eqref{eq:constcond} requiring ${\widehat \alpha}_k$ to be constant is reasonable 
in some physical systems. 
For example, linear amplifiers will have this property over their intended operational bandwidth,
see \cite{whit2002}.
However, some physical processes do not have this property. In particular, signal strength loss 
from free space propagation in both the radio frequency and acoustic regimes is highly 
dependent on frequency, e.g., \cite{rapp2009}. High frequencies are more attenuated than low frequencies,
and this means that the ${\widehat \alpha}_k(f), f \in \{1, \ldots, M\}$, values may vary.
Hence, the DFT example in Section \ref{sec:datasim} will not satisfy 
condition \eqref{eq:constcond}.
On the other hand, if the matrix $A$ of Proposition \ref{prop:sep1} is sufficiently simple, then
condition \eqref{eq:constcond} is not necessary to ensure that
$S$ is separable, as the following result shows.
\end{rem}

\begin{prop}[Linear separable sensing scenario -- matrix constraint]
\label{prop:sep2}
Let $S$ be a pre-separable sensing scenario.
Let $H$ be a linear mapping defined by
\[
       H: {\mathbb C}^M \rightarrow {\mathbb C}^n,
\]
\[
       H(x) = Ax,
\]
where A is an $n\times M$ matrix of complex numbers. Suppose the rows of $A$ are 
constant multiples of  rows taken from the rows of the $M \times M $ identity matrix. 
Then, $S$ is a separable sensing scenario and $H$ is a health space mapping in the sense
of Definition \ref{defn:sep}.
\end{prop}
\begin{proof}
The $i$-th row of $A$ consists of $M-1$ zeros and one non-zero entry at a position we 
shall call $f_i$, that is, $a_{i,f}=0$ for all  $f \ne f_i$ and $a_{i,f_i} \ne 0$. Therefore, 
for each $j = 1, \ldots, N$ and $k = 1, \ldots, K$, we can apply $H$ as follows:
\[
     H(v_{j,k})(i) = \sum_{f=1}^M a_{i,f} v_{j,k}(f) = a_{i,f_i} v_{j,k}(f_i) = 
     a_{i,f_i}{\widehat \gamma}_j(f_i) {\widehat \alpha}_k(f_i) = \gamma_j(i) \alpha_k(i),
\] 
where we set $\gamma_j(i)=a_{i,f_i}{\widehat \gamma}_j(f_i)$ 
and $\alpha_k(i)={\widehat \alpha}_k(f_i)$. This completes the proof.
\end{proof}

\begin{rem}[Connection with DFT example]
\label{rem:sep2dft}
Propositions \ref{prop:sep1} and \ref{prop:sep2}
are relevant for showing that the 
mathematical techniques developed here will be applicable to the DFT 
scenario of Example \ref{ex:dft} 
and
Section \ref{sec:datasim}. 
In particular, Proposition \ref{prop:sep2} shows that the projection mapping
described in Section \ref{sec:datasim} 
is a health space mapping for a separable sensing scenario. Further, the proof of 
Proposition \ref{prop:sep1} shows that certain practicalities for dealing with real systems will not cause 
problems. For example, spectral lines that fall between DFT bins might require weighted averaging of 
adjacent values of the DFT output. Since these values correspond to adjacent frequencies, we can 
assume ${\widehat \alpha}_k(f) \approx {\widehat \alpha}_k(f+1),$ and the proof of 
Proposition \ref{prop:sep1} assures separability.
\end{rem}


\subsection{Radiative, dominant, and harmonious separable sensing scenario}
\label{sec:raddomharm}

Definition \ref{defn:sep} allows us to quantify certain natural properties of the sensing
scenario. In particular, it allows us to distinguish between a failed sensor and a failed
component. It also allows us to quantify the notion of a sensor bearing primary responsibility
for a particular parameter and for a parameter to be heard by more than one sensor.
In order to justify these claims, we make the following definitions.

\begin{defn}[Radiative and dominant separable sensing scenario]
\label{defn:raddom}
Let $S$ be a separable sensing scenario, and let $i \in \{1, \ldots, n\}.$

{\it a.} $S$ is {\it i-radiative} if
\[
 \exists \,k_i \in \{1, \ldots, K\} \; \text{ such that } \; {\alpha}_{k_i}(i) \neq 0.
\]

{\it b.} $S$  is {\it i-dominant} if
\[
       \exists \,j_i \in \{1, \ldots, N\} \; \text{ such that } \; {\gamma}_{j_i}(i) \neq 0;
\]
and it is {\it strongly i-dominant} if
\[
       \exists \,j_i \in \{1, \ldots, N\} \text{ such that } \forall \ell \neq j_i, \hskip.2in 
       |{\gamma}_{j_i}(i)| > (N - 1)|{\gamma}_{\ell}(i)|,
\]
where $N$ is the number of sensors.
\end{defn}

\begin{rem}[Rationale for radiativity and dominance]
\label{rem:rad}
{\it a.} We have chosen the word, {\it radiative}, from the notion that in order for a
parameter to be sensed, e.g., for a spectral line of a DFT to register at a given
sensor, the object must radiate some sort of energy, viz., the $\widehat{{\alpha}}s$. Similarly, 
we have chosen the word, {\it dominant}, since the sensor which bears primary responsibility 
for reporting a parameter
 must be able to {\it see/hear} it loudly, viz., the $\widehat{{\gamma}}s.$ Naturally, if the
object stops sending out energy, no sensor could detect it. Thus,
in the case of engines (see Section \ref{sec:datasim}),
radiativity has to do with producing noise while dominance has to do with 
hearing the noise at the sensor.

{\it b.} Under a mild condition about the existence of non-zero elements, we have that
{\it if a sensor bears primary responsibility for a parameter, then the given
separable sensing scenario $S$ is $i$-dominant for some $i$.}
To see this, suppose we are given a sensor, $s_j$,  
that bears primary responsibility 
for a parameter, $f \in \{1,\ldots,M\}$, i.e., $f \in S_j$. Then,
consider the $1$-dimensional subspace ${\mathbb C}_{f}
=  {\mathbb C} \subseteq {\mathbb C}^M$ 
generated by the canonical basis vector for ${\mathbb C}^M$ with a $1$
in the $f$ position.
Since $H : {\mathbb C}^M \longrightarrow {\mathbb C}^n$ 
and ${\mathbb C}_{f}
\subseteq {\mathbb C}^M$, we have that
$H({\mathbb C}_{f})$ is a subset of ${\mathbb C}^n$.
Therefore, if we {\it assume} further that $H({\mathbb C}_{f})$
has a non-zero element $H(v_{j,k}) \in {\mathbb C}^n,$ then 
\[
     \exists i \in \{1,\ldots,n\} \;  \text{such that} \;  H(v_{j,k})(i) \ne 0;
\] 
and so
\[
     \gamma_j(i) \ne 0,
\] 
by the definition of the health space mapping $H$.
Consequently, with $j=j_i$, we see from Definition \ref{defn:raddom} that
$S$ is {\it i-dominant} for this $i$.

{\it c.} Of course, as the name implies, we would usually like to choose the $j_i$ in the definition of $i$-dominance 
so that $\gamma_{j_i} (i)$ is maximal in some sense. 
However, there are instances, particularly when there are noise considerations in the sensor output, 
when this may not be the case. For example,  the largest $\gamma_{j_i}$ may also be significantly 
noisier, and possibly have a lower SNR, than another choice. 
\end{rem}

\begin{figure}[htbp]
\begin{center}
\includegraphics[height=2.5in]{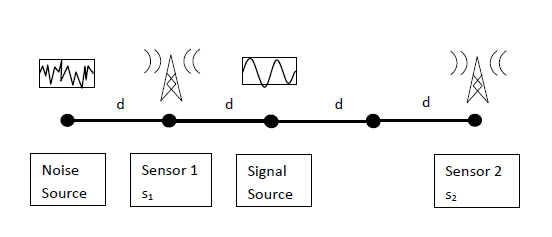}
\caption{Dominance vs SNR}
\label{fig:domsnr}
\end{center}
\end{figure}

\begin{example}[Dominance vs SNR]
\label{ex:snrdominance}
Recall that SNR is defined as the ratio of the power of the desired
signal to the background noise power, generally measured on a logarithmic
scale in terms of  decibels (dB), i.e.,
\[
      {\rm SNR} = {\rm P}_{{\rm signal}}/{\rm P}_{{\rm noise}} = {\rm P}_{{\rm signal},{\rm dB}} - {\rm P}_{{\rm noise},{\rm dB}},
\]
where ${\rm P}_{{\rm signal}, {\rm dB}} =10 \, {\rm log}_{10}({\rm P}_{{\rm signal}})$ and  ${\rm P}_{{\rm noise}, {\rm dB}} =
10 \, {\rm log}({\rm P}_{\rm noise})$ for signal and noise power, respectively.

As an example of the effect of SNR considerations on dominance, consider a situation, where there are 
2 sensors, a signal source and a noise source, arranged in a line as in Figure \ref{fig:domsnr}. The first sensor, 
$s_1,$ is located at a distance $d$ from the noise source and the signal source. The second sensor, 
$s_2$, is located at a distance $2d$ from the signal but $4d$ from the noise. This setup works equally 
well for both acoustic and RF sensing
modeling, as long as we assume that the signal and noise propagate 
according to free space propagation, see \cite{rapp2009}.
According to such propagation, if the distance between 
the source and the
sensor doubles, the received power level will decline by a factor of 4. Thus, if the signal strength at $s_1$ 
is ${\bf s}$, then the signal strength of $s_2$ will be ${\bf s}/4$. Hence, if signal strength were the only determining factor, 
then
$\gamma_1$ would be maximal, and consequently it would be the obvious choice for a dominant value. 
However, 
if the signal strength of the noise at $s_1$ is ${\bf N}$, then, according to free space propagation, the noise 
power at $s_2$ will be ${\bf N}/16$. Therefore,
if we compute SNR values, then the SNR at $s_1$ is ${\bf s}/{\bf N},$ while at $s_2$ it is
$({\bf s}/4)/({\bf N}/16) = 4{\bf s}/{\bf N}$. Thus, if SNR is a consideration, it is not unreasonable to choose 
$\gamma_2$ 
as the dominant value.
\end{example}

\begin{example}[Primary responsibility]
\label{ex:dftdominance}
The canonical example of {\it primary responsibility}
is given by the DFT scenario described in Sections \ref{sec:data} and 
\ref{sec:sim}. There, the health space mapping, {\it H}, simply picks out the {\it n} significant 
frequencies from the {\it M}-point DFT. Thus, {\it H} is a projection mapping.
In this case of the DFT,
the parameters for which a sensor bears primary responsibility are the 
frequencies  in $\{1, \ldots, M\}$ 
for which that sensor gives significant information.
These frequencies 
will be mapped by the projection mapping onto the non-zero elements of ${\mathbb C}^n$. In fact, this 
defines an explicit mapping of {\it only} the significant frequencies, $f \in \{1,...,M\},$ 
mapping into ${\mathbb C}^n.$
\end{example}

We want to quantify the idea described in Example \ref{ex:dftdominance}. 
In particular,
given a separable sensing scenario with health space mapping {\it H} and health space 
${\mathbb C}^n$, we wish to allow ${\mathbb C}^n$ to inherit the primary responsibility attributes 
available in ${\mathbb C}^M$. To this end, in Definition \ref{defn:subsets} we shall define
notation for various parameters and sets of indices that are needed to quantify dominance,
see Definition \ref{defn:op}, 
and then to define harmonious scenarios, see Definition 
\ref{defn:disjharm}.

\begin{defn}[Subsets of $\{1, \ldots, n\}$ for a separable sensing scenario]
\label{defn:subsets}

Let $S$ be a separable sensing scenario.  

{\it a.} Consider the vectors, $H(v_{j,k}) \in {\mathbb C}^n,$ 
where $1 \leq j \leq N$ and $1 \leq k \leq K$. For each $j = 1, \ldots, N,$
construct the set
\[
J_j = \{ i \in \{1, \ldots, n\} : \exists k = k(i) \in \{1, \ldots, K\}      \;\text{such that} \, H(v_{j,k})(i) \ne 0 \},
\]
see Figure \ref{fig:Jsets}.

A subset $J_j$ could be the empty set. Without loss of
generality,
we can assume that each $i \in \{1, \ldots, n\}$ can be found in at least one such
set of indices. If not, then for that $i$,
\[
\forall j, k, \quad H(v_{j,k})(i)=0,
\]
in which case the
health space mapping could be improved by projecting down to ${\mathbb C}^{n-1}$,
eliminating the index $i$ altogether.

Further, we could have $J_j \cap J_{j'} \neq \emptyset$ for some $j \neq j'.$
We note here that the $J_j$ sets form a covering of the set $\{1,...,n\}$ in much the same way as
the sets $T_j$ form a covering of $S$, see the beginning of Subsection \ref{sec:sepsensscen}. 

{\it b.} We now wish
to construct sets, $I_j$, of indices  that correspond to the partition $\{S_j\}$. $I_j$ will have cardinality
${\rm card}\,(I_j)=n_j$, and we would like $\sum_{j=1}^{N} n_j =n$.
As such, for each $j = 1, \ldots, N,$ we
choose sets $I_j = \{ i_{\ell,j}: \ell =1,..., n_{j} \} \subseteq J_j, \,
I_j \subseteq J_j,$ with the properties that
$\{I_j\}$ is a disjoint collection and $\bigcup_{j=1}^N \{ I_{j} \} = \{1,...,n\},$ 
see Figure \ref{fig:Isets}. We note that some
of the $I_j$ can be the empty set even when $J_j \neq \emptyset$; in this case $n_j = 0$.

Summarizing, we have that
\begin{equation}
\label{eq:ijell}
 \forall i \in \{ 1,\ldots,n \}, \; \exists j=j(i) \in \{ 1,\ldots, N\} \, \text{and} \, \ell \in \{1, \ldots, n_j\}
  \, \text{such that} \, i \in I_j \, {\rm and} \,
   i = i_{\ell,j},
 \end{equation}
and
\begin{equation}
\label{eq:ik}
   \forall i = i_{\ell,j} \in I_j, \, \exists k = k(i) \in \{1,\ldots, K\} \, \text{such that} \, H(v_{j,k})(i_{\ell,j}) \neq 0.
\end{equation}

\end{defn}

\begin{figure}[htbp]
\centering
\includegraphics[height=0.22\paperheight]{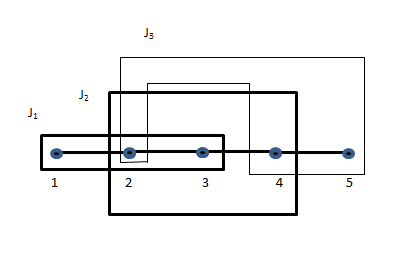}
\caption{Sets $J_1$, $J_2$, $J_3$}
\label{fig:Jsets}
\end{figure}

\begin{figure}[htbp]
\centering
\includegraphics[height=0.22\paperheight]{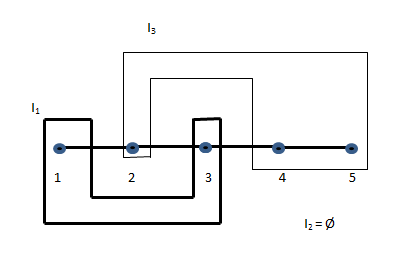}
\caption{Sets $I_1$, $I_2$, $I_3$}
\label{fig:Isets}
\end{figure}

\begin{example}[Subsets of $\{1, \ldots, n\}$ for a separable sensing scenario]
\label{ex:JI}
We illustrate the relationship between the sets $I_j$ and $J_j$, see Figures
\ref{fig:Jsets} and \ref{fig:Isets}.
Consider a case where $n=5$ and $N=3$ We could have sets $J_1=\{1,2,3\}$,
$J_2=\{2,3,4\}$, and $J_3=\{2,4,5\}$.
We may then choose $I_1=\{1,3\}$, $I_2= \emptyset$, and $I_3=\{2,4,5\}$, in
which case we have $n_1=2$,
$n_2=0$, and $n_3=3$, so that $n_1+n_2+n_3=5=n$.
\end{example}

\begin{defn}[Operational sensors]
\label{defn:op}
Given the set-up of Definition \ref{defn:subsets},
we have that $H(v_{j,k})(i) \ne 0$ for $i \in I_j$ and so 
$\gamma_j(i) \ne 0$ as well. In this case we shall say that the sensor $s_j$ is {\it operational}. 
If something goes wrong with the sensor, some or all of the $\gamma$s will be affected. 
We say that $s_j$ is {\it non-operational} if

\[
\exists \ell \in \{1,\ldots,n_j\} \neq \emptyset \quad \text{such that} \quad
\gamma_j(i_{\ell,j}) = 0.
\]
\end{defn}

\begin{rem}
\label{rem:primrespnonop}
{\it a.} [$I_j$ and primary responsibility]
There is a connection between the $I_j$ 
of Definition \ref{defn:subsets} and the notion of primary responsibility. We 
noted in Remark \ref{rem:rad}b that if a sensor $s_j$ bears primary responsibility for a parameter,
$f$, then the scenario is $i$-dominant for some $i =i(j,f)$. In particular, we showed the existence 
of an $i$ for which $H(v_{j,k})(i) \ne 0$, where $k$ arises in Remark \ref{rem:rad}b.
We would like this $i$ to be
an element of  $I_j$. 
However, this may not be possible since there may be
another parameter, $f'$, for which a sensor $s_{j'}$ 
bears primary responsibility and for which $f'$ gives rise to the same value of $i$.


 
{\it b.} [Non-operational sensors]
Clearly, the failure of a sensor, i.e., a sensor becoming non-operational, could have a dramatic affect on the 
sensing scenario. In particular, if sensor $s_j$ becomes non-operational, then $\gamma_j(i_{\ell,j}) = 0$ implies 
that the scenario may no longer be $i$-dominant for $i=i_{\ell,j}$. Indeed, this will precisely be the case if there 
is no $j'\ne j$ for which $i_{\ell,j} \in J_{j'}$. Alternatively, however, it may be that there does exist such a $j'$, 
see Example \ref{ex:radar} and Figure \ref{fig:radar}. In this example we can consider RADAR B to 
play the role of sensor $s_j$. We think of the area directly around B as corresponding to indices which are not 
in any other $J_{j'}$ while the area around the X corresponds to indices which are also in the $J$ set associated
 with RADAR A.

We note that since $H$ operates on the $v_{j,k}$, should we find that
$H(v_{j,k})(i)=\alpha_k(i)\gamma_j(i)= 0$, then we have no way of telling if it is because $\alpha_k(i)=0$ or 
because $\gamma_j(i)=0$. This can be a crucial difference. In a RADAR/SONAR example, 
$\alpha_k=0$ may simply 
correspond to a lack of targets in the area, and may not be an unusual event, while a sensor becoming 
non-operational may lead to a disastrous security breach. On the other hand, for the turbine/DFT example, 
the opposite is true: a sensor failure would not affect proper operation of the machine while the failure of a 
turbine might have much more serious consequences. Section \ref{sec:datasim} explores this 
possibility in detail.

\end{rem}

\begin{defn}[Harmonious separable sensing scenario]
\label{defn:disjharm}
Let $S$ be a separable sensing scenario, and given the partition of indices described above. $S$ is {\it j-disjoint} if
\[
\exists \ell \in \{1,...,n_j \} \quad \text{ such that } \quad \forall j'\ne j, \quad {\gamma}_{j'}(i_{\ell,j})=0.
\]
If $S$ is not $j$-disjoint then we say $S$ is {\it j-harmonious}, and, in this case,
\[
     \forall \ell \in \{1,...,n_j \}, \, \exists\, j' \ne j \quad \text{ such that }  \quad \gamma_{j'}(i_{\ell,j})\,\neq \, 0.
\]
$S$ is
{\it harmonious} if it is not $j$-disjoint for any $j,$ that is, it is $j$-harmonious for each $j.$
\end{defn}

\begin{rem}[Rationale for harmony]
\label{rmk:disj}
We have chosen the words, {\it harmonious} and {\it disjoint}, for the following
reasons. Multiple sensors can 
often sense the same parameters, hence the term, harmonious. 
On the other hand, there may be a parameter that can only be sensed
by the dominant sensor, and so we use the term, disjoint, in this case. 

We wish to exploit
the property of being harmonious in the following way. In the event of sensor failure, a harmonious scenario 
may be able to recover some information about parameters even if its primary sensor 
is the one that failed. On the other hand, a scenario that is $j$-disjoint for all $j$ 
will not be able to recover any information about a parameter if its dominant sensor fails.
\end{rem}

\begin{example}[RF radio: radiativity, dominance, and harmony]
\label{ex:specsenscen}
Consider a pair of RF radios receiving signals from two transmitters at different frequencies. 
Let the distance from the first transmitter to the first receiver be $d_1$ and the distance to the 
second receiver be $4d_1$. Let the distance from the second transmitter to the second receiver 
be 
$d_2$ and the distance to the first receiver be $4d_2$. If we assume free space propagation in a 
non-fading environment, then each receiver should hear two signals, one $12 dB$ down from the 
other. See Figure \ref{fig:trans-rec}, where, for simplicity and concreteness, we have arranged the sensors
and transmitters linearly and have taken 
$d = d_1 = d_2$.

\begin{figure}[htbp]
\centering
\includegraphics[height=0.22\paperheight]{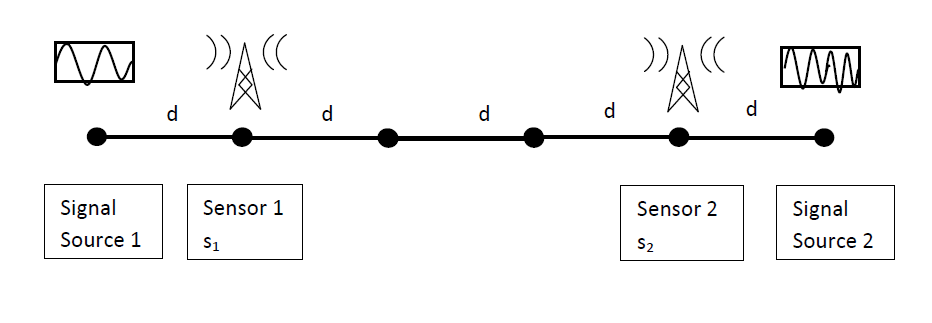}
\caption{Transmitters and receivers}
\label{fig:trans-rec}
\end{figure}

This situation describes a {\it separable sensing scenario} since the received signal level on each 
frequency at each receiver is given as the transmitted power, $\alpha_1$ and $\alpha_2,$
respectively, multiplied by 
the free space losses, $\gamma_1(1),$ $\gamma_1(2),$ $\gamma_2(1),$ and $\gamma_2(2),$ at each receiver. 
The scenario is {\it radiative} if both transmitters are operating. If one of the transmitters is turned off, 
it ceases to be radiative. It is $j$-{\it dominant} for $j = 1,2$; in fact, since $N=2,$ it is 
{\it strongly} $j$-{\it dominant}. It is also
{\it harmonious} since both receivers hear both transmitters.
 
Now consider adding a third transmitter-receiver pair at some great distance and with substantial 
blocking, e.g., on the other side of a mountain or, for an extreme case, the other side of a 
continent. 
Transmitter 3 cannot be received by either of the first two receivers nor can the third receiver 
hear 
either of the first two transmitters.  The scenario is still separable, although several of the 
$\gamma_j$ 
will be zero; and still radiative provided all transmitters are turned on. 
However, it is no longer harmonious 
since it is $3$-disjoint. 
\end{example}

\begin{example}[SONAR and RADAR: harmony]
\label{ex:sonradharm}
The SONAR sensing scenario described in Example \ref{ex:sonar} can be considered harmonious, 
assuming the area of concern is to the north. Much of the area covered by sensor A is 
also covered by sensor B, albeit with greater noise problems since some areas will only be 
covered by endfire beams. We note that this scenario is not harmonious if the area
of concern is to the east or west.

The RADAR scenario of Example \ref{ex:radar} will not be harmonious over the oblong sensing scenario described in 
Figure \ref{fig:radar}. We note, however, that partial recovery of information is possible using noisy bins near the 
edge of coverage. 
\end{example}


\section{Bases, frames, and multiplicative frames}
\label{sec:frames}

 \subsection{Frames}
 \label{sec:frame}
 
 There is an intimate connection between reactive sensing theory and the theory of frames.
 In this subsection we shall define frames, and state some of their relevant properties.
 The theory of frames 
 will be used in Section \ref{sec:rst} 
 as a natural tool for the analysis of reactive sensing problems. 
 Such problems led to our concept of multiplicative frames 
 defined in Subsection \ref{sec:multframe}.

A {\it frame} for ${\mathbb C}^d$ is a sequence, $X = \{x_h\}_{h=1}^D \subseteq
{\mathbb C}^d,$ that spans ${\mathbb C}^d,$ i.e.,
\begin{equation}
\label{eq:span}
    \forall \,x \in {\mathbb C}^d, \; \exists\,c_1,\dots, c_D\in {\mathbb C}^d
    \; {\rm such \; that} \quad x = \sum_{h=1}^D\, c_h\,x_h.
\end{equation}
This innocent and elementary property is the basis (sic) for the
power of frames, and it belies the
power of finite frames in dealing with numerical stability, robust signal representation,
and noise reduction problems,
see, e.g., \cite{daub1992}, \cite{BenFra1994} Chapters 3 and 7,
\cite{chri2016}, \cite{KovChe2007a}, and \cite{KovChe2007b}.
The following definition 
for Hilbert spaces is equivalent to the definition of frames for ${\mathbb C}^d,$
but is formulated in terms of bounds that are often useful in computation
and coding.

\begin{defn}[Frames]
\label{defn:frame}
{\it a.} Let $\HH$ be a separable Hilbert space over the field $\FF$,
where $\FF = \RR$ or $\FF = \CC,$ 
e.g., 
${\mathbb H} = L^2({\mathbb R}^d), {\mathbb R}^d, {\mathbb C}^d.$
A finite or countably infinite sequence, 
$X = \{x_h\}_{h \in J},$ of elements of $\HH$ 
is a {\it frame} for $\HH$ if 
\begin{equation}
\label{eq:frame}
      \exists A, B > 0  \; \text{such that}  \;
      \forall x \in \HH, \quad A\norm{x}^2 \leq \sum_{h \in J} |\langle {x},{x_h}\rangle|^2 \leq B\norm{x}^2.
\end{equation}
The optimal constants, viz., the supremum over all such $A$ and infimum over all such $B$, are 
called the {\it lower} and {\it upper frame bounds} respectively. When we refer to {\it frame bounds} 
$A$ and $B$, we shall mean these optimal constants. Otherwise, we use the terminology,
{\it a lower frame bound} or {\it an upper frame bound}.

{\it b.} A frame $X$ for ${\mathbb H}$ is a {\it tight frame} if $A = B.$ If a tight frame has the further property 
that $A = B = 1,$ then the frame is a {\it Parseval frame} for ${\mathbb H}.$  

{\it c.}  
A tight frame $X$ for ${\mathbb H}$ is a {\it unit norm tight frame} 
if each of the elements 
of $X$ has norm $1.$  
Finite unit norm tight frames for finite dimensional ${\mathbb H}$
are designated as FUNTFs.

{\it d.} A sequence of elements of ${\mathbb H}$,
not necessarily a frame, satisfying an upper frame bound,
such as $B\norm{x}^2$ in (\ref{eq:frame}), is a {\it Bessel sequence}.

{\it e.} Let ${\mathbb V}$ be a vector space over $\mathbb C.$ A sequence, $X = \{x_h\}_{h=1}^D 
\subseteq {\mathbb V}$,
is a {\it basis} for ${\mathbb V}$ if it spans ${\mathbb V}$ as in Equation \eqref{eq:span} and if 
$\{x_h\}_{h=1}^D$ is a linearly independent set, in which case $D$ is the {\it dimension}
of ${\mathbb V}.$ ${\mathbb V}$ is {\it infinite dimensional} if it contains an infinite linearly independent set.
Clearly, if ${\mathbb V}$ is a separable Hilbert space, then every basis for ${\mathbb V}$ is a 
frame for ${\mathbb V}$.
\end{defn}



Let $X = \{x_h\}_{h \in J}$ be a frame for ${\mathbb H}$. We define 
the following operators associated with every frame; they 
are crucial to frame theory. The {\it analysis operator} $L : {\mathbb H} \rightarrow \ell^2(J)$ is defined by
\[
       \forall x \in {\mathbb H},  \quad Lx = \{\inner{x}{x_h} \}_{h \in J}.
\]
The adjoint of the analysis operator is the {\it synthesis operator} 
$L^\ast : \ell^2(J) \rightarrow {\mathbb H}$, and it is defined by
\[
        \forall a \in \ell^2(J),  \quad L^\ast a= \sum_{h \in J} a_h x_h.
\]
The {\it frame operator} is the mapping $\Fcal : {\mathbb H} \rightarrow {\mathbb H}$ defined as 
$\Fcal = L^\ast L$, i.e., 
\[
\forall x \in {\mathbb H}, \quad \Fcal(x) = \sum_{h \in J} \inner{x}{x_h}  x_h.
\]

The following is a fundamental theorem.

\begin{thm}[Frame reconstruction formula]
\label{thm:framerecon}
Let ${\mathbb H}$ be a separable Hilbert space, and let $X = \{x_h\}_{h \in J} \subseteq {\mathbb H}$. 

{\it a.} $X$ is
 a frame for ${\mathbb H}$ with frame bounds $A$ and $B$ if and only if 
 $\Fcal : {\mathbb H} \rightarrow {\mathbb H}$ is a topological isomorphism with norm bounds
 $\norm{\Fcal}_{op} \leq B$ and $\norm{\Fcal}_{op}^{-1} \leq A^{-1}$.  
 
 {\it b.} In the case of either condition of part {\it a}, we have the following:
 \begin{equation}\label{eq:frameopinverse}
B^{-1}I \leq \Fcal^{-1} \leq A^{-1} I,
\end{equation}
$\{ \Fcal^{-1} x_h \}$ is a frame for $H$ with frame bounds $B^{-1}$ and $A^{-1}$, and 
\begin{equation}
\label{eqn:recon}
        \forall x \in {\mathbb H}, \quad x = \sum_{h \in J} \inner{x}{x_h}\Fcal^{-1} x_h = \sum_{h \in J} \inner{x}{\Fcal^{-1}x_h} x_h
        = \sum_{h \in J} \inner{x}{\Fcal^{-1/2}x_h}\Fcal^{-1/2}x_h.
\end{equation}

\end{thm}

For a proof of part {\it a.}, see \cite{BenWal1994}, pages 100--104.
 
 For part {\it b.},
let $X = \{x_h\}_{h \in J}$ be a frame for $H$. Then, the frame operator
$\Fcal$ is invertible (\cite{daub1992}, \cite{bene1994}); and $\Fcal$ is a multiple 
of the identity precisely when $X$ is a tight frame. 
Further, $\Fcal^{-1}$ is a positive self-adjoint operator and has a square root $\Fcal^{-1/2}$ 
(Theorem 12.33 in \cite{rudi1991}). This square root can be written as a power series in $\Fcal^{-1}$; 
consequently, it commutes with every operator 
that commutes with $\Fcal^{-1},$ and, in particular, with $\Fcal.$ These properties allow us to assert that 
$\{ \Fcal^{-1/2} \, x_h \}$ is a Parseval frame for $\HH$, and give the third equality of \eqref{eqn:recon}.
see \cite{chri2016}, page 155.

\begin{rem}[Frames and bases for $\CC^d$ and $\HH$]
In light of 
the fact that orthonormal bases (ONBs) are 
frames, it is natural to ask to what extent frames can be constructed in terms of
ONBs. This is pertinent because of our frame results in Section \ref{sec:rst} and our simulation in
Section \ref{sec:datasim}. 
\begin{itemize}
\item It may  be considered surprising that any infinite dimensional $\HH$ contains a frame for $\HH$
which does not contain a basis for $\HH$. The result is due to Casazza and Christensen, see \cite{chri2016}, 
Chapter 7, for details.
\item The first result relating frames and sums of bases is due to Casazza \cite{casa1998}.
{\it 
Let $\HH$ be a separable Hilbert space over the field $\FF$, and let $X =\{ x_h \}_{h \in J}$ be a frame for $\HH$
with upper frame bound $B$. Then, for every $\epsilon > 0$, there are ONBs $\{u_h\}_{h \in J},
\{v_h\}_{h \in J}, \{w_h\}_{h \in J}$ for $\HH$ and a constant $C =B(1+\epsilon)$ such that
}
\[
        \forall h \in J,\quad x_h = C(u_h + v_h + w_h).
\]
The proof depends on an operator-theoretic argument.
\end{itemize}
\end{rem}

\subsection{Multiplicative frames}
\label{sec:multframe}

Because of the formulation described in Section \ref{sec:mathmodel}, we define
the notion of a multiplicative frame.

\begin{defn}[Multiplicative frames]
\label{defn:mf}
A sequence, $X = \{x_{j,k}\} \subseteq {\mathbb C}^n, 1 \leq j \leq N, 1 \leq k \leq K$ is a
{\it multiplicative frame} for ${\mathbb C}^n$ if it is a frame for ${\mathbb C}^n$ and
if
\[
   \exists \,\{y_j : 1 \leq j \leq N\} \subseteq {\mathbb C}^n \: \text{ and } \: \exists\,\{z_k : 1 \leq k \leq K\} 
   \subseteq {\mathbb C}^n
   \quad \text{ such that }
\]
\[
\forall\, j = 1, \ldots, N, \, \forall \,k = 1, \ldots, K, \: \text{ and } \forall \, i = 1, \ldots, n, \quad
   x_{j,k}(i) = y_j(i) \, z_k(i).
\]
\end{defn}

In the following constructions of multiplicative frames, it is important to note
that they require the hypotheses of radiativity and dominance
given in Definition \ref{defn:raddom}. In fact, radiativity is manifested by 
\eqref{eq:zposmin} and dominance is manifested by \eqref{eq:yposmin}
in Theorem \ref{prop:mf1}. Besides Corollary \ref{prop:mf2}, this should also be
compared with Theorem \ref{thm:frame}.
We view this as a striking connection between the mathematical concept of
a multiplicative frame and the notions of radiativity and dominance that
arise from our physical modeling of a sensing scenario.

\begin{thm}[A construction of multiplicative frames]
\label{prop:mf1}
Given $Y = \{y_j : 1 \leq j \leq N\} \subseteq {\mathbb C}^n
\setminus{\{0\}} \: \text{ and } \: Z = 
\{z_k : 1 \leq k \leq K\} \subseteq {\mathbb C}^n\setminus{\{0\}}$, and define
 $X = \{x_{j,k} : x_{j,k}(i)=y_j(i) \, z_k(i)\} 
\subseteq {\mathbb C}^n, 1 \leq j \leq N, 1 \leq k \leq K, 1 \leq i \leq n$. 

{\it a.} Let $Y$ be a frame for ${\mathbb C}^n$ with frame constants $A_N$ and $B_N$, and 
assume
\begin{equation}
\label{eq:zposmin}
    \exists \, k \in \{1, \ldots, K\} \; {\rm such \, that} \; {\rm min}_{1\leq i \leq n}\, |z_k(i)| = m_z > 0.
\end{equation}
Then, $X$ is a multiplicative frame for ${\mathbb C}^n$; and an upper frame bound B and  
a lower frame
bound A are constructed in the proof.

{\it b.}  Let $Z$ be a frame for ${\mathbb C}^n$ with frame constants $A_K$ and $B_K$, and 
assume
\begin{equation}
\label{eq:yposmin}
    \exists \, j \in \{1, \ldots, N\} \; {\rm such \, that} \; {\rm min}_{1\leq i \leq n}\, |y_j(i)| = m_y > 0.
\end{equation}
Then, $X$ is a multiplicative frame for ${\mathbb C}^n$; and an upper frame bound B and  
a lower frame
bound A are constructed in the proof.

\end{thm}
\begin{proof} {\it a.} First, ${\rm card}\,\{x_{j,k}\} = NK$. By the multiplicative definition of 
each $x_{j,k}$, it is sufficient to prove that
$X$ is a frame for ${\mathbb C}^n$. 

Let $x = (x(1), x(2), \ldots, x(n)) \in {\mathbb C}^n$.
Then,
\[
   \sum_{j=1,\ldots,N, \,k=1,\dots,K}\,|\langle x, x_{j,k}\rangle |^2 
   = \sum_{k=1}^K \sum_{j=1}^N \, |\langle x, x_{j,k}\rangle |^2
   =  \sum_{k=1}^K \sum_{j=1}^N \,|\langle x, y_jz_k\rangle |^2
\]
\[
     = \sum_{k=1}^K \sum_{j=1}^N \,|\langle x \overline{z}_k, y_j\rangle |^2 
      = \sum_{k=1}^K \sum_{j=1}^N \, \left|\sum_{i=1}^n\,x(i)\overline{z_k(i)}\,\overline{y_j(i)}\right|^2.
\]
Thus, since
\[
  \sum_{k=1}^K \sum_{j=1}^N \, |\langle x, x_{j,k}\rangle |^2 
  = \sum_{k=1}^K\,\sum_{j=1}^N \,|\langle x \overline{z}_k, y_j\rangle |^2,
\]
we have
\begin{equation}
\label{eq:multframeinproof}
   \sum_ {k=1}^K\,A_N\,\| x \overline{z}_k \|_2^2 \leq 
   \sum_{k=1}^K\,\sum_{j=1}^N\,|\langle x, x_{j,k}\rangle |^2  \leq
    \sum_ {k=1}^K\,B_N\,  \| x \overline{z}_k \|_2^2.
\end{equation}

Next, we make the estimate,
\begin{equation}
\label{eq:multnormeq}
    \sum_ {k=1}^K\,\| x \overline{z}_k \|_2^2  =
   \sum_ {k=1}^K\,\left( \sum_{i=1}^n\,|x(i)|^2|\overline{z_k(i)}|^2\right)
\end{equation}
\[  
     \leq \sum_{k=1}^K\,{\rm max}_{1 \leq i  \leq n}\,|z_k(i)|^2\, \left(\sum_{i=1}^n\,|x(i)|^2\right)
\]  
\[
     = \|x\|_2^2\, \sum_{k=1}^K\, \|z_k\|_{\infty}^2 
     \leq K\,\left({\rm max}_{1 \leq k \leq K}\, \|z_k\|_{\infty}^2 \right)\,\|x\|_2^2.
\]
Therefore,
\[
     \sum_{j=1,\ldots,N, \,k=1,\dots,K}\,|\langle x, x_{j,k}\rangle |^2 
     \leq K\,B_N\, \left({\rm max}_{1 \leq k \leq K}\, \|z_k\|_{\infty}^2 \right)\,\|x\|_2^2.
\]
Consequently, and only assuming that $\{y_j\}$ is a Bessel sequence with 
Bessel bound $B_N$, we obtain that
$\{x_{j,k}\}$ is a Bessel sequence for ${\mathbb C}^n$ with Bessel bound,
and an upper frame bound,
$ K\,B_N\,\left({\rm max}_{1 \leq k \leq K}\, \|z_k\|_{\infty}^2\right)$.

Finally, to obtain a   lower frame bound, we proceed as follows.
We combine \eqref{eq:multframeinproof} and \eqref{eq:multnormeq} with the
hypothesis, \eqref{eq:zposmin}, to make the estimate,
\[
      \sum_ {k=1}^K\,A_N\,\| x \overline{z}_k \|_2^2 \geq 
      m_z\,A_N\,\|x\|_2^2.
\]
In particular, a lower frame bound is $m_z\,A_N$.

{\it b.} The proof of part {\it b} is analogous to that of part {\it a} with the
roles of $Y$ and $Z$ reversed.
  \end{proof} 

We state Corollary \ref{prop:mf2} as a corollary of Theorem \ref{prop:mf1}.
In fact, since none of the elements of $Y$ or $Z$ is the $0$-vector, conditions
\eqref{eq:zposmin} and \eqref{eq:yposmin} are automatically satisfied.

\begin{cor}[A construction of multiplicative frames for two frames]
\label{prop:mf2}
Let $Y = \{y_j : 1 \leq j \leq N\} \subseteq {\mathbb C}^n \setminus \{0 \}
\: \text{ and } \: Z = 
\{z_k : 1 \leq k \leq K\} \subseteq {\mathbb C}^n \setminus \{0\}$
be frames for ${\mathbb C}^n$ Then,  $X = \{x_{j,k} : x_{j,k}(i)=y_j(i) \, z_k(i)\} 
\subseteq {\mathbb C}^n, 1 \leq j \leq N, 1 \leq k \leq K, 1 \leq i \leq n$ is a multiplicative 
frame for ${\mathbb C}^n.$
\end{cor}

In our forthcoming theory of multiplicative frames 
we would need to evaluate  
more refined frame bounds for results
such as Theorem \ref{prop:mf1} and  Corollary \ref{prop:mf2}. 

Our intention is to connect the notion of a separable sensing scenario to the theory of frames. 
In Section \ref{sec:bmfm} we shall start with a separable sensing scenario and use the health space 
mapping $H$ to construct basis and frame mappings from sets $\{v_{j,k}\}$ of sensor data into ${\mathbb C}^n,$
where $j$ indicates one of the $N$ sensors, $k$ designates one of the $K$ times, and each $v_{j,k}$ is evaluated
at some $f \in \{1, \ldots, M\}.$
In some cases,
these mappings will give rise to {\it multiplicative bases} for ${\mathbb C}^n,$ while in other cases they will give rise 
to {\it multiplicative frames}. The former will be called {\it basis mappings}, while the latter will be called 
{\it frame mappings}, as defined in Subsection \ref{sec:bmfm}. 

\section{Reactive sensing theory}
\label{sec:rst}

\subsection{Basis and frame mappings}
\label{sec:bmfm}


Consider a separable sensing scenario $S$ with health space ${\mathbb C}^n$  and 
health space mapping $ H : {\mathbb C}^M \longrightarrow {\mathbb C}^n$. 
We wish to analyze the state of $S$. At time 
$k, 1 \leq k \leq K,$ each sensor $s_j, 1 \leq j \leq N,$ generates a vector 
$v_{j,k} \in {\mathbb C}^M$, see the beginning of Subsection \ref{sec:sepsensscen}
for notation.
We use $H$ to map the 
$v_{j,k}$ to $H(v_{j,k}) \in {\mathbb C}^n.$ Then, generally, we shall have some detection
and/or classification scheme in place in ${\mathbb C}^n,$ that allows us to say something about the 
state of $S$ given the 
information $H(v_{j,k}) \in {\mathbb C}^n$, see Subsection \ref{sec:detection}.

For example, in the DFT Example \ref{ex:dft} and later in Section \ref{sec:datasim}, it may be the case 
that if a particular coordinate $H(v_{j,k})(i)$ is greater than some known value, then a 
gear fault is indicated.


\begin{rem}[The role of frames]
\label{rem:frames}
One issue that arises immediately is that for a fixed $i \in \{1, \ldots, n\},$ there could be many 
$H(v_{j,k})(i)$ values from 
which to choose at time $k$. In fact,  we may have one 
such value for each $j$, $1 \leq j \leq N;$ and $N$ could be large. There are 
applications where a 
solution may require deploying hundreds of inexpensive 
sensors, see, e.g., \cite{LeeKriKuo2004}, \cite{Wu-LiuKin2005}, \cite{boye2009}.

It is this overdetermined nature of reactive sensing problems that led us to use the theory of 
frames. In particular, we note that frames are over-complete sets of atoms,
as opposed to bases, where the loss of just one basis element can permanently and adversely affect
accurate signal representation in terms of the remaining elements of the basis. 
In this regard, and considering
the first signal reconstruction equality of Equation (\ref{eqn:recon}) in Section \ref{sec:frames}, 
the fact that the $x_h$ 
may form an over-complete
set of atoms 
raises the possibility that the $NK$ vectors $H(v_{j,k}) \in {\mathbb C}^n$ could be a frame
for ${\mathbb C}^n,$ see Theorems \ref{thm:frame1} and \ref{thm:frame}. 
\end{rem}

We now define frame and basis mappings 
 in terms of the mathematical model of Section \ref{sec:mathmodel}. 
To set the stage, we consider simultaneously all of the sensor output values in 
${\mathbb C}^{MN}$ at a fixed time $k, \, 1 \leq k \leq K$. 
Our frame and basis mappings are the means of transferring
this data to health space, ${\mathbb C}^n$, where $n$ should typically be less than $MN$. Recall that
\[
      {\rm card}\,\{v_{j,k}\} = KN, \quad {\rm card}\,(S) = M, \quad {\rm and} \quad S= \bigcup_{j=1}^N T_j.
\]
Formally, we define the set, 
\begin{equation}
\label{eq:V}
      V=\{v_{1,k}\oplus v_{2,k}\oplus \ldots  \oplus v_{N,k} : k = 1 \leq k \leq K\}  \subseteq {\mathbb C}^{MN},
      \quad {\rm card}\,(V) =  K,
\end{equation}
consisting of the $K$ $MN-$tuples, $(v_{1,k}(f_1), \ldots, v_{1,k}(f_M), \ldots, v_{N,k}(f_1), \ldots, v_{N,k}(f_M)),$
and we define the mapping,
 \[
      {\mathcal H} : {\mathbb C}^{MN} \longrightarrow {\mathbb C}^{nN},
 \]
 by applying $H$ to each appropriate copy of ${\mathbb C}^{M}$ in ${\mathbb C}^{MN}$ so that  
\[
{\mathcal H} (v_{1,k}\oplus v_{2,k}\oplus \ldots \oplus v_{N,k})= 
H(v_{1,k})\oplus H(v_{2,k})\oplus \ldots \oplus H(v_{N,k}) \in {\mathbb C}^{nN}.
\]

As noted in Subsection \ref{sec:theme}, Definition \ref{defn:BFmappings}b
is technical, and necessarily so in order to obtain the effective quantitative results of 
Section \ref{sec:datasim}. However, it really only reflects a computational means for using 
the over-completeness inherent in frames. Definition \ref{defn:BFmappings}a is also 
technical, and is an analogue for bases in order to compare the
roles of bases and frames in reactive sensing implementations.

\begin{defn}[Basis and frame mappings]
\label{defn:BFmappings}
Let $S$ be a separable sensing scenario with health space mapping $H$.

{\it a.i.} The first way to describe the state of $S$ at a fixed time, $k,\,
k=1,\ldots, K$, is to analyze the values assigned to parameters in $S_j$ by
sensors $s_j,$ that bear {\it primary responsibility}
 for those parameters
at time $k$, see the beginning of Subsection \ref{sec:sepsensscen}. This method 
will ignore the values assigned to $T_j\diagdown S_j$ by
sensor $s_j,$ and now allows us to define a mapping,
\begin{equation}
\label{eq:vu}
        v_{j,k} \mapsto u_{j,k},
\end{equation}
where $v_{j,k} \in {\mathbb C}^M$ and $u_{j,k} \in {\mathbb C}^n$,
and where we shall give precise meaning to $u_{j,k}$ for a given $v_{j,k}$.

{\it a.ii.} To this end, we begin with 
a separable sensing scenario $S$ and the sets $I_j$ defined in
Definition \ref{defn:subsets} of Subsection \ref{sec:raddomharm}.
Recall that, notationally,
\[
  \forall i \in \{ 1,\ldots,n \}, \; \exists j=j(i) \in \{ 1,\ldots, N\}, \, n_j, \, {\rm and} \,
  \ell \in \{1, \ldots, n_j\}  \, 
  {\rm such} \, {\rm that} \, i = i_{\ell,j},
\]
where $n = n_1 + \cdots + n_N$.
We define the {\it basis mapping}, 
\[
     B: {\mathbb C}^{MN} \longrightarrow {\mathbb C}^n,
\]
by the formula
\begin{equation}
\label{eq:basis}
        B( \bigoplus_{j'=1}^N c_{j'}) (i_{\ell,j}) = H(c_j)(i_{\ell,j}) \in {\mathbb C},
\end{equation}
where 
$j \in \{1, \ldots, N\}$, $H(c_j) \in {\mathbb C}^n$ by definition of the 
given health space mapping $H$ associated with $S$,
$\ell \in \{1, \ldots, n_j\}$, and each $c_{j'} \in {\mathbb C}^M$, $1 \leq j' \leq N$.
Then, if we fix $j \in \{1, \ldots, N\}$ and $k \in \{1, \ldots, K\}$, and set $c_{j'} = 0 \in 
{\mathbb C}^M$ for each $j' \neq j$, we write
\[
            B( (0, \ldots, 0, v_{j,k}, 0, \ldots, 0)) = u_{j,k} \in {\mathbb C}^n. 
\]
Thus, for each $j \in \{1, \ldots, N\}$ and $k \in \{1, \ldots, K\}$, we have
\begin{equation}
\label{eq:ujk}
   u_{j,k} = 
   \begin{cases} 
     H(v_{j,k})(i)  & \text{if} \, i = i_{\ell,j} \, \text{for some} \, \ell \in \{1,\ldots, n_j\},     \\
     0   &  \text{if} \,  i = i_{\ell, j'} \, \text{for all}\,  j' \neq j.
   \end{cases}
\end{equation}


{\it a.iii.} As a simple, illustrative example, consider a scenario with two sensors, $s_1$ and $s_2$, each of which 
produces a two dimensional vector for each time $k$. Suppose $s_1$ bears primary responsibility for the 
first coordinate, while $s_2$ bears primary responsibility for the second. Let us assume $H$ is the identity mapping, 
so that ${\mathbb C}^M = {\mathbb C}^n = {\mathbb C}^2$. Suppose at time $k$, $v_{1,k}=(10,2)$
and $v_{2,k}=(-1,7)$. Then, by \eqref{eq:ujk},
we have $u_{1,k}=(10,0)$ and $u_{2,k}= (0,7)$. We also have that one element 
of $V$, defined by Equation (\ref{eq:V}), will be $(10,2) \oplus (-1,7) \in {\mathbb C}^{2 \cdot 2} =  {\mathbb C}^4$.
 Applying the mapping $B$ will pick 
out the first element of the first vector and the second element of the second to form the result: 
$B( (10,2) \oplus (-1,7))= (10,7)$.

{\it a.iv.} We note that $B$ naturally decomposes into a 
composition. In fact, each element of $ {\mathbb C}^{MN}$ is mapped into
its own copy of health space by means of the $  {\mathcal H}$ mapping,
followed by the mapping,
$B'$, that uses \eqref{eq:basis} to select the 
values from each image $H(c_j) \in \CC^n$ in order to construct the resulting vector. Thus, we write
$$ 
        B = B' \circ {\mathcal H}. 
$$

{\it b.i.} The second way to describe the state of $S$ at a fixed time $k$ is to analyze all of the 
reported data
supplied by {\it all} of the sensors. Thus, rather than only selecting data from the sensor that bears 
primary responsibility for a parameter, we combine the data from {\it all} of the sensors. 
This approach allows us to define
a mapping,
\begin{equation}
\label{eq:vw}
         v_{j,k} \mapsto w_{j,k},
\end{equation}
where $v_{j,k} \in {\mathbb C}^M$, and where $w_{j,k} \in {\mathbb C}^n$ 
is only necessarily zero for components corresponding to
parameters for which $s_j$ does not report any values at all. 
The result of this strategy is to define what we call
{\it frame mappings}, and we shall give precise meaning to $w_{j,k}$ for a given $v_{j,k}$
in \eqref{eq:Fw}.


{\it b.ii.} Unlike the situation for basis mappings described in part {\it a.ii},
there are many possible frame mappings that can be formulated. For concreteness 
we shall define one, in particular, that will exhibit properties that are most useful 
for frame mappings in general. We define the {\it magnitude sum frame mapping},
\[
     F: {\mathbb C}^{MN} \longrightarrow {\mathbb C}^n,
\]  
by the formula 
\begin{equation}
\label{eq:msfm}
        F( \bigoplus_{j'=1}^N c_{j'}) (i) = \sum_{j'=1}^N |H( c_{j'})(i)| \in {\mathbb C}, \quad 1 \leq i \leq n,
\end{equation}
where
$j' \in \{1, \ldots, N\}$, $H(c_{j'}) \in {\mathbb C}^n$ by definition of the 
given health space mapping $H$ associated with $S$,
and each $c_{j'} \in {\mathbb C}^M$, $1 \leq j' \leq N$.
Then, if we fix $j \in \{1, \ldots, N\}$ and $k \in \{1, \ldots, K\}$, and set $c_{j'} = 0 \in 
{\mathbb C}^M$ for each $j' \neq j$, we write
\begin{equation}
\label{eq:Fw}
        F( (0, \ldots, 0, v_{j,k}, 0, \ldots, 0)) = w_{j,k}, \; {\rm where} \, {\rm each} \; 0 \in {\mathbb C}^M,
\end{equation}
taking into account that $H(0) = 0 \in \CC^n$, see Definition \ref{defn:sep}.

Thus, for each $j = 1, \ldots, N, \, k= 1, \ldots, K$, and $i = 1, \dots, n$, 
and using the fact that $H(0) = 0$, we have
\begin{equation}
\label{eq:Hvw}
         w_{j,k}(i) = |H(v_{j,k})(i)|.
\end{equation}
 To see this
consider the following $N$-tuple of $M$-tuples for \eqref{eq:msfm}:
\[
     c_{j'} = 0 \in \CC^M \; {\rm if} \; j' \neq j \quad {\rm and} \quad c_{j'} = v_{j,k} \in \CC^M \; {\rm if} \; j' = j.
\]
Then,
\[
     F( \bigoplus_{j'=1}^N c_{j'}) (i) =  |H(v_{j,k})(i)|,
\]
and so \eqref{eq:Hvw} is obtained by \eqref{eq:Fw}.

{\it b.iii.} For the numerical example in part {\it a.iii}, the mapping, $v_{j,k} \mapsto w_{j,k}$,
is simply the magnitude sum frame mapping
on the components, and so $w_{1,k} = (10,2)$ and $w_{2,k} = (1,7)$. Evaluating $F$, we 
compute that
$F( (10,2)\oplus (-1,7))= (10,2) + (1,7) = (11,9)$.

{\it b.iv.} We note that $F$ naturally decomposes into a 
composition. In fact, each element of $ {\mathbb C}^{MN}$ is mapped into
its own copy of health space by means of the $  {\mathcal H}$ mapping,
followed by the mapping, 
$F'$, that is {\it defined} by computing the sum of the magnitudes of the components of each 
$H(c_j) \in \CC^n$ in order to construct the resulting vector, see \eqref{eq:msfm}. 
Thus, we write
\[ 
    F = F' \circ {\mathcal H}.
\]

{\it c.}
The basis mapping $B$ is linear and the magnitude sum frame
mapping $F$ is non-linear, although it could be linear on subspaces.
In general,
linearity is neither necessarily natural nor desirable for some
realistic sensing scenarios. For example, if the signal to noise ratio
(SNR) of data reported by different sensors varies, it may make sense
to scale data non-linearly before performing the addition in the frame
mapping $F$, see Subsection \ref{sec:snr}.

\end{defn}

\begin{rem}[A non-commutative diagram]
\label{rem:noncomm}
There is the following natural non-commutative diagram, Figure \ref{fig:noncom},
associated with Definition \ref{defn:BFmappings}.

\begin{figure}[!ht]
\centering
\begin{equation*}
\begin{CD}
 V\subset \CC^{MN} @>{\mathcal H}>> \CC^{nN}\\
@VV{\mathcal H}V @VVB'V\\
\CC^{nN} @>F'>> \CC^n
\end{CD}
\end{equation*}
\caption{Non-commutative diagram}
\label{fig:noncom}
\end{figure}

In order to illustrate the two flows, basis and frame, in Figure \ref{fig:noncom} and the 
precise definition of both $B' \circ {\mathcal H}$ and $F' \circ {\mathcal H}$,
recall that $\{I_j\}$ (Definition \ref{defn:subsets}) is a {\it partition} of
$\{1, \ldots, n\}$ and that some of the $I_j$ could be empty. Thus, if $i \in \{1, \ldots, n\}$,
then $i$ is in a unique $I_j$ and there is a unique $\ell \in \{1, \ldots, n_j\}$ such that
$i = i_{\ell,j}$, see ($\ref{eq:ijell}$) and ($\ref{eq:ik}$). Also, from the definition of 
${\mathcal H}$ we have
\[
          {\mathcal H} (v_k)= 
         H(v_{1,k})\oplus H(v_{2,k})\oplus \ldots \oplus H(v_{N,k}) \in {\mathbb C}^{nN},
\]
where $v_k = v_{1,k}\oplus v_{2,k}\oplus \ldots \oplus v_{N,k} \in V$.

Thus, for the case of basis mappings and for a fixed $i \in \{1, \ldots, n\}$, we
compute for any $k$, that
\[
         B' \circ {\mathcal H}(v_k)(i) = H(v_{j,k})(i_{\ell,j}) \in \CC,
\]
where $j$ and $\ell$ are specified by the initial choice of $i$.
Similarly, for the case of magnitude sum frame mappings and for a fixed $i \in \{1, \ldots, n\}$, we
compute for any $k$, that
\[
          F' \circ {\mathcal H}(v_k)(i) = \sum_{j'=1}^N |H(v_{j,k})(i_{\ell,j})| \in \CC.
\]

For the example of Definition \ref{defn:BFmappings} {\it a.iii}, {\it b.iii}, we note that
\[
  B' \circ {\mathcal H}\,((10,2) \oplus (-1,7)  = (10,7), 
\]
and
\[
  F' \circ {\mathcal H}\,((10,2) \oplus (-1,7) = (11,9).
\]

\end{rem}

\begin{rem}[An advantage of frame mappings]
\label{rem:frameadvantage}

Specific basis and frame mappings, $B$ and $F$, can be
constructed 
depending on the application.
In fact, the complexity of the mappings, $B$ and $F$, can be expected to 
be data and/or problem dependent, and there are many possible definitions.
However, in all cases, we want to distinguish quantitatively between the
singular role of a single sensor or small group of sensors associated with a component
of a sensing scenario (the mappings $B$) and the impact of {\it all}
of the sensors (the mappings $F$) to deal with the case that
a particular sensor or 
group of sensors might be disabled.

An advantage of a frame mapping is that in the event of the loss of one or 
more sensors, some of the missing data can be recovered. For example, suppose 
$s_j$ bears primary responsibility for a parameter $f$. As formulated in 
Remarks \ref{rem:rad}b and \ref{rem:primrespnonop}a, suppose there is an $i \in I_j$
associated with $f$. If sensor $s_j$ becomes 
non-operational, then the basis mapping vectors $u_{j,k}$ are all $0$. 
Further, by the formulation of a basis mapping, all of the other basis 
mapping vectors, $u_{j',k}, \, j'\ne j$, will have the property that $u_{j',k}(i) = 0$.
Thus, all vectors in the image of the basis
mapping will have their $i$-th component equal to $0$;
and hence vital information has been lost. However, for the 
frame mapping, there may be a vector $w_{j',k},\, j' \ne j$ for which $w_{j',k}(i) \ne 0$. 
This will be the case, in particular, when the scenario is $j$-harmonious.
We  
shall explore this phenomenon in Subsection \ref{sec:bt}.

\end{rem}



\subsection{Fundamental theorems}
\label{sec:bt}

In this subsection we prove the main results of reactive sensing theory.
We shall show that if a  basis mapping is applied to an appropriate set, then it will 
generate a basis, and if a frame mapping is applied to the same set, 
then it will generate a frame. These facts are essential in order for our mathematical
model to distinguish between sensor failure and critical sensed events.

For the basic set-up, consider the set $V$, which contains
the $K$ vectors, $v_{1,k}\oplus \ldots \oplus v_{j,k} \oplus \ldots \oplus v_{N,k}
\in V \subseteq {\mathbb C}^{MN}, \, 1\leq k \leq K$.  Each $v_{j,k}$ 
is mapped by $H$ to $H(v_{j,k}) \in {\mathbb C}^n.$ Using the notation from Subsection \ref{sec:bmfm},
we write
${\mathcal H}:V \rightarrow {\mathbb C}^{nN}$ evaluated at  
$v_{1,k}\oplus \ldots \oplus v_{j,k} \oplus \ldots \oplus v_{N,k}$ as
\[
 {\mathcal H}(v_{1,k}\oplus \ldots \oplus v_{j,k} \oplus \ldots \oplus v_{N,k}) 
 = (H(v_{1,j}),...,H(v_{j,k}),...,H(v_{N,k})) \in {\mathbb C}^{nN},
 \] 
where each $H(v_{j,k})$ is a vector of the form,
\[
(H(v_{j,k})(1),...,H(v_{j,k})(i),...,H(v_{j,k})(n)) \in  {\mathbb C}^n.
\]

We now look at the individual components for the image vectors under this mapping. 
Specifically, we define the {\it projective} set $X$ as 
\begin{equation}
\label{eq:projframe}
       X = \{ (0,...,0,H(v_{j,k})(i),0,...,0) \in {\mathbb C}^{n} : 1\le i \le n,\  1\le j\le N,\  1\le k \le K\},
\end{equation}
where $H(v_{j,k})(i)$ is the $i$-th component of the vector $H(v_{j,k})$
in $\CC^n.$ Note that
${\rm card}\,X = nNK$.

\begin{thm}[Conditions for projective multiplicative frames]
\label{thm:frame1}
Let $S$ be a separable sensing scenario with partition $\{S_j\}_{j=0}^N,$ 
covering $\{T_j\}_{j=1}^N,$ and mappings $s_j:T_j \times \{1, \ldots ,K\} \rightarrow \CC^M$, where 
$ s_j(f,k) = v_{j,k}(f) $  for $f \in T_j$ and $ v_{j,k}(f) = 0 $ for $ f \in S \backslash T_j .$  Let
$H : {\mathbb C}^M \rightarrow {\mathbb C}^n$ be a health space mapping of the form, 
\[
      \forall\, j = 1,\ldots, N \;{\rm and}\; \forall \,k = 1,\ldots, K, \quad H(v_{j,k})(i) = 
      {\gamma}_{j}(i) \, {\alpha}_{k}(i),
\]
where $i = 1,\ldots, n.$ 
If $S$ is $i-$radiative and $i-$dominant
for each $i = 1,\ldots ,n,$ then
$X \subseteq \CC^n,$ defined by Equation \eqref{eq:projframe}, is a multiplicative frame for $ \CC^n.$
\end{thm}

\begin{proof} Take any $i \in \{1,...,n\},$ noting it is of the form $i=i_{\ell,j}$ for some 
$j \in \{1,...,N\}$ and some $\ell \in \{1,...,n_j\}$, see Definition \ref{defn:subsets}.

From radiativity, for $i=i_{\ell,j},$ there is $k_i$ such that $\alpha_{k_i}(i) \ne 0$. 
Since $i=i_{\ell,j}$ and $s_j$ is $i$-dominant, we have $\gamma_j(i_{\ell,j}) \ne 0.$ Therefore,
$H(v_{j,k_i})(i_{\ell,j}) \ne 0$. Consequently, 
\[
       (0,\dots,0, H(v_{j,k_i})(i),0,\dots,0) \in X \subseteq \CC^n
\] 
is not the zero vector. Hence, the $n$ vectors we obtain, one for each $i,$ form a basis for $\CC^n$ 
and the result is proved.
\end{proof}

\begin{thm}[Harmonious multiplicative frames]
\label{thm:frame2}
Let $S$ be a separable sensing scenario as described in Theorem \ref{thm:frame1}.
Assume $S$ is $i-$radiative and $i-$dominant
for each $i = 1,\ldots ,n$, and assume $S$ is  $j$-harmonious. If sensor $s_j$ fails 
and becomes non-operational, 
then $X \subseteq \CC^n $ is still a multiplicative frame for $ \CC^n.$
\end{thm}

\begin{proof} Fix $j$. Since $s_j$ is {\it non-operational} we have that
\[
\exists \ell \in \{1,\ldots,n_j\} \neq \emptyset \quad \text{such that} \quad
\gamma_j(i_{\ell,j}) = 0.
\]
Then, as before,
\[
      \forall i=i_{\ell,j'},\  \ell \in \{1,...,n_{j'}\},\  j'\ne j, \quad  
     (0,\dots,0, H(v_{j,k_i})(i),0,\dots,0) \ne 0 \in \CC^n.  
\] 
This gives $n-n_j$ non-zero basis vectors.

For $i=i_{\ell,j},\ \ell \in \{1,...,n_j\}$, we may have $\gamma_j(i_{\ell,j})=0.$ 
However, by the $j-$harmonious property,
\[
        \exists j'_{\ell,j} \quad \text{such that} \quad\gamma_{j'_{\ell,j}}(i_{\ell,j}) \ne 0.
\] 
From $i$-radiativity, 
there is $k_i\ \text{such that} \ \alpha_{k_i}(i_{\ell,j}) \ne 0$. Therefore, we have
\[
       H(v_{j'_{\ell,j},k_i})(i_{\ell,j}) = \gamma_{j'_{\ell,j}}(i_{\ell,j})\alpha_{k_i}(i_{\ell,j}) \ne 0.
\] 
Thus,
\[
     (0,...,0,H(v_{j'_{\ell,j},k_i})(i_{\ell,j}),0,...,0) \ne 0 \in \CC^n.
\]
This gives the remaining $n_j$
non-zero basis vectors, thereby creating a basis in $\CC^n,$ and this completes the proof.
 \end{proof}

We now relate Theorems \ref{thm:frame1} and \ref{thm:frame2} to the basis 
and frame mappings defined in Subsection \ref{sec:bmfm}.

As above, $ V = \{v_{1,k}\oplus v_{2,k}\oplus ... \oplus v_{N,k} : k = 1,...,K\} \subseteq {\mathbb C}^{MN},$ 
and we have $B(V)\subseteq {\mathbb C}^{n}$ and $F(V)\subseteq {\mathbb C}^{n},$ as well as
the compositions,
$ B = B' \circ {\mathcal H}$ and $F = F' \circ {\mathcal H}.$ We want to examine the component
mappings, $B'$ and $F'$. To this end we look at the coordinate components of the elements of
 ${\mathcal H}(V) \subseteq{\mathbb C}^{nN},$ 
obtained by projecting each image vector onto its $nN$ coordinates. Thus, for each vector
$v \in V, v=v_{1,k}\oplus v_{2,k}\oplus ... \oplus v_{N,k} $ for some $k$, we obtain $nN$ 
{\it projection} vectors of the form
$(0,0,0...,0,(0,0,...,0,H(v_{j,k})(i),0,...,0),0,...,0)$. Note that the 0s inside only one set of parentheses 
are the 
$0\in {\mathbb C}^n$ vectors, while the 0s in the inner set of parentheses are $0 \in {\mathbb C}$. 
Further, the
cardinality of the set of all such projection vectors 
is $nNK$. We wish to look at a subset of these vectors obtained by 
taking radiativity into account.
Assuming we are working with a radiative separable sensing scenario, then, for each $i,\ i=1,...,n$, 
we can find $k=k_i$,
such that $\alpha_{k_i}(i) \neq 0$. 
As such, we define
\[
        Z =\{ (0,0,..,0,(0,0,...0,H(v_{j,k_i})(i),0,...,0),0...,0) , j=1,...,N, \text{ and } i=1,...,n\} \subseteq
        \CC^{nN},
\]
where, once again, the 0s inside only one set of parentheses 
are the 
$0\in {\mathbb C}^n$ vectors, while the 0s in the inner set of parentheses are $0 \in {\mathbb C}$. 
Note that card(Z)=$nN$. We now apply the mappings $B'$ and $F'$ to the set $Z$.






\begin{thm}[Basis mappings and a basis for $\CC^n$]
\label{thm:basismapping}
Let $S$ be a separable sensing scenario as described in Theorem \ref{thm:frame1}.
Assume that $S$ is $i-$radiative and $i-$dominant for each $i = 1,\ldots ,n$.
The set, $B'(Z)$, is the union of a basis for $\CC^n$ and the vector,
$0\in\CC^n$.
\end{thm}

\begin{proof}
Let $z_{j,i}= (0,...,0,(0,...,0,H(v_{j,k_i})(i),0,...,0),0,...,0)$. Then, from the definition of $B$ and $B',$ we have
$B'(z_{j,i}) =0 \in \CC^n $ unless $i=i_{\ell,j}$.
In this case we have $B'(z_{j,i})(i') =0 \in \CC $ except when $ i'=i=i_{\ell,j}$. $B'(z_{j,i})(i) = \gamma_j(i) \alpha_{k_i}(i) \neq 0 $
since $i=i_{\ell,j}$. This gives precisely $n$ non-zero vectors one for each $i_{\ell,j}$, 
each a different canonical basis vector.
\end{proof}

\begin{cor}[Non-operational-sensors and non-bases]
\label{cor:basismapping}
Let $S$ be a separable sensing scenario as described in Theorem \ref{thm:frame1}.
Assume $S$ is $i-$radiative and $i-$dominant for each $i = 1,\ldots ,n$. 
Let
$s_j$ be a sensor with $n_j \neq 0$. If $s_j$ is non-operational,
then $B'(Z)$ will no longer span ${\mathbb C}^n$ and, hence, will no longer contain a basis.
\end{cor}

\begin{proof}
$s_j$ non-operational implies there exists an $\ell, \ell \in \{1,\ldots,n_j\} \neq \emptyset$ with
$\gamma_j(i_{\ell,j}) = 0$. Then, from the proof of Theorem \ref{thm:basismapping} the set of $n$ vectors given at
the end contains $B'(z_{j,i})$ which is of the form $(0,0,0,...,H(v_{j,k})(i),0,...,0) = 0 \in {\mathbb C}^n$ since
$H(v_{j,k})(i) = \gamma_j(i) \, \alpha_k(i) =0 $. Thus, there cannot be more than $n-1$ non-zero vectors in $B'(Z)$
and it cannot span ${\mathbb C}^n$
\end{proof}

The situation for the frame mapping is more complicated since there can be many possible frame mappings.
The next result deals with the magnitude sum frame mapping, but similar results can be 
obtained for other
frame mappings.

\begin{thm}[Frame mappings and multiplicative frames for $\CC^n$]
\label{thm:framemapping}
Let $S$ be a separable sensing scenario as described in Theorem \ref{thm:frame1}.
Assume $S$ is $i-$radiative and $i-$dominant for each $i = 1,\ldots ,n$. 
Consider the magnitude sum frame mapping, $F$.The set, $F'(Z)$, 
contains a multiplicative frame for $\CC^n$.
\end{thm}

\begin{proof}
For each $i = 1, \ldots, n$ and each $j = 1, \ldots, N$, consider the vector,
\[
     z_{j,i}= (0,...,0,(0,...,0,H(v_{j,k_i})(i),0,...,0),0,...,0) \in Z \subseteq \CC^{nN},
\] 
which we can define by the $i$-radiativity.
Then, from the definitions of $F$ and $F',$ we see that
\[
F'(z_{j,i})(i')=
    \begin{cases}
|\gamma_j(i) |\, |\alpha_{k_i}(i)|  & \text{ if $i'=i$,} \\
 0 & \text{ otherwise}.
     \end{cases}
 \]
Thus, we have 
\[
 F'(z_{j,i}) = (0,...,0,|H(v_{j,k_i})(i)|,0,...,0)=(0,...,0, |\gamma_j(i)|\, | \alpha_{k_i}(i)|,0,...,0) \in \CC^n.
\]
Unlike the basis case, given $j$,  $F'(z_{j,i}) $ may be non-zero even if 
$i \notin I_j$, i.e., even if
$i \neq i_{\ell,j}$ for any $1 \leq \ell \leq n_j$ .
However, for each $i$, we do have $i=i_{\ell,j}$ for some $\ell$ and $j$, so in that case
$F'(z_{j,i})(i) = |\gamma_j(i)|\, | \alpha_{k_i}(i)| \neq 0 $. Thus, for each $i \in 1,\ldots,n$ we have at least one
non-zero vector which is only non-zero in the $i$-th position. Choosing one of these vectors for each $i$, we form a 
basis contained in $F'(Z)$ and
therefore $F'(Z)$ contains a frame. Since this is a separable sensing scenario, the frame vectors satisfy Definition \ref{defn:mf} and the frame is multiplicative.
\end{proof}

\begin{rem}[$F'(Z)$ and $B'(Z)$]
\label{rem:framemappingproofproperties}
We note that $F'(Z)$ and $B'(Z)$ are different sets. First, ${\rm card}\,B'(Z) = n+1$ since $B'(Z)$ consists 
of $n$ basis vectors plus the zero vector. The set,
$F'(Z)$, on the other hand, may or may not contain the zero vector, 
since it is possible that $\gamma_j(i) \ne 0$ for all $i \in \{1,...,n\}$ and for all $j \in \{1,...,N\}$. 
However,
$F'(Z)$ could contain several similar basis vectors depending on how many $j'$ have the property that
$\gamma_{j'}(i_{\ell,j}) \ne 0$ for a given $j$.
If $\gamma_j(i) \neq 0$ for some $i \neq i_{\ell,j}$ and any $\ell$, then there will be duplicate
vectors with non-zero $i$-th components. As the following
corollary shows, this property turns out to be useful if a sensor fails

\end{rem}

\begin{cor}[Non-operational sensors and multiplicative frames]
\label{cor:framemapping}
Let $S$ be a separable sensing scenario as described in Theorem \ref{thm:frame1}.
Assume $S$ is $i-$radiative and $i-$dominant for each $i = 1,\ldots ,n$;
and also assume that $S$ is j-harmonious as in Theorem \ref{thm:frame2}.
Consider the magnitde sum frame mapping, $F$.
Then, even if sensor $s_j$ 
is non-operational, the set, $\{F'(Z)\} $,
will still contain a multiplicative frame for $\CC^n.$
\end{cor}

\begin{proof}
Following the pattern of Theorems \ref{thm:frame2} and \ref{thm:framemapping}, we let
\[
      z_{j,i}= (0,...,0,(0,..., 0, H(v_{j,k_i})(i), 0,...,0), 0,...,0)
\] for all $j.$ As in Theorem \ref{thm:framemapping} we have
$$F'(z_{j,i})(i')=\begin{cases}
|\gamma_j(i) |\, |\alpha_{k_i}(i)|  & \text{ if $i'=i$} \\
 0 & \text{ otherwise},
\end{cases}$$
and 
\[ 
F'(z_{j,i}) = (0,...,0,|H(v_{j,k_i})(i)|,0,...,0)=(0,...,0, |\gamma_j(i)|\, | \alpha_{k_i}(i)|,0,...,0) \in \CC^n.
\]
Now, fix a $j$ and assume $s_j$ is non-operational. Then,
\[
 \forall i=i_{\ell,j'},\ \ell \in \{1,...,n_{j'}\},\ j'\ne j, \ \ \ F'(z_{j',{k_i}}) \ne 0 \in \CC^n
\]
as before. This gives $n-n_j$ non-zero basis vectors.
For $i=i_{\ell,j},\ \ell \in \{1,...,n_j\}$ we may have $\gamma_j(i_{\ell,j})=0.$ However, by $j$-harmony,
\[
      \exists j'_{\ell,j} \quad {\rm such \, that} \quad \gamma_{j'}(i_{\ell,j}) \ne 0.
\] 
Therefore, we have
\[
     H(v_{j'_{\ell,j},k_i})(i_{\ell,j}) = \gamma_{j'}(i_{\ell,j})\alpha_{k_i}(i_{\ell,j}) \ne 0,
\] 
and so
\[
      F'(z_{j',i}) = (0,\ldots,0,|H(v_{j'_{\ell,j},k_i})(i_{\ell,j})|,0,\ldots,0) \neq 0 \in \CC^n.
\]
This gives the remaining $n_j$
non-zero basis vectors, one for each $i = i_{\ell,j}, \, \ell \in \{ 1, \ldots, n_j\}$.
Thus, we have constructed a basis for $\CC^n,$ in $F'(Z)$. Thus, $F'(Z)$ contains
a frame. Since this is a separable sensing scenario, the frame vectors satisfy Definition \ref{defn:mf} 
and the frame is multiplicative.
\end{proof}


\subsection{A general theorem and constructive technique for multiplicative frames}
\label{sec:furtherres}

Theorem \ref{thm:frame} below should be compared with Theorem \ref{thm:frame1}. To this end, recall
that the vectors $w_{j,k} \in {\mathbb C}^n$ are defined as
\begin{equation}
\label{eq:wjk}
       w_{j,k}= F(0,...0, v_{j,k}, 0,... 0),
\end{equation}
where, for each fixed $k = 1, \ldots, K$, $F$ is the magnitude sum frame mapping defined on 
${\mathbb C}^{MN}$. Because of this, it makes sense to construct
multiplicative frames directly from the set $V$ defined in \eqref{eq:V}
and illustrated in Figure \ref{fig:noncom}.
Also, recall that the frame $X$ of
Theorem \ref{thm:frame1} has $nNK$ elements, whereas the frame,
$\{ w_{j,k} \} \subseteq {\mathbb C}^n$, of  Theorem \ref{thm:frame}
has $NK$ elements. Note that we assume $n \leq N$ in Theorem
\ref{thm:frame}, cf. Remark \ref{rem:health}{\it e}. This assumption is needed in our
proof of Theorem \ref{thm:frame}, but there are examples of sensing scenarios
where $\{ w_{j,k} \}$forms a frame even when $N < n$.

\begin{thm}[Strong dominance and multiplicative frames]
\label{thm:frame}
Let $S$ be a separable sensing scenario with partition $\{S_j\}_{j=1}^N,$ 
covering $\{T_j\}_{j=1}^N,$
and mappings $s_j:T_j \times \{1, \ldots ,K\} \rightarrow \CC^M$, where 
$ s_j(f,k) = v_{j,k}(f) $  for $f \in T_j$ and $ v_{j,k}(f) = 0 $ for $ f \in S \backslash T_j .$ 
Let $H : {\mathbb C}^M \rightarrow {\mathbb C}^n$ be a health space mapping of the form, 
$$
      \forall\, j = 1,\ldots, N \;{\rm and}\; \forall \,k = 1,\ldots, K, \quad H(v_{j,k})(i) = 
      {\gamma}_{j}(i) \, {\alpha}_{k}(i),
$$
where $i = 1,\ldots, n.$ Assume $N >  1$ and $n \leq N$. Consider the magnitude sum frame mapping, $F$.
If $S$ is $i-$radiative and strongly $i-$dominant
for each $i = 1,\ldots ,n,$ then
$ \{w_{j,k}:j=1,\ldots ,N \,{\rm and} \, k=1, \ldots ,K\} \subseteq \CC^n $ defined by
$F$ in Equation \eqref{eq:wjk} is a multiplicative frame for $ \CC^n.$
\end{thm}

\begin{proof} {\it i.} We shall calculate that 
$ \{w_{j,k}:j=1,\ldots ,N \;{\rm and} \; k=1, \ldots ,K\}$ contains a basis for $\CC^n,$ 
and this proves the result. 
The calculation is contained in 
part {\it ii}, where we use $i-$radiativity, and in parts {\it iii} and {\it iv}, where we use
strong $i-$dominance to make a basic estimate.

{\it ii.}
Separability allows us to write $w_{j,k}=\gamma_j\alpha_k,$ where $\gamma_j,\alpha_k \in \CC^n$ 
for each $j=1, \ldots, N$ and $k=1, \ldots, K$. Clearly, ${\rm card}\{w_{j,k}\} = NK.$

Fix $i \in \{1, \ldots,n\}$. 
Since $S$ is $i-$radiative, there is a $ 1 \le k_i \le K$ such that $\alpha_{k_i}(i) \neq 0.$
Therefore, we can choose $k_i' \in \{1, \ldots, K\}$ such that 
\begin{equation}
\label{eqn:thma}
\forall \, k=1, \ldots, K, \quad |\alpha_{k_i'}(i)| \ge |\alpha_k(i)| \quad {\rm and} \quad |{\alpha}_{k_{i}^{'} }(i)| > 0.
\end{equation}
Equation (\ref{eqn:thma}) follows by choosing ${k_i'},$ among all possible $k_i,$ that gives the largest value 
of $|\alpha_{k_i}(i)|$.

Taking any $1 \le j \le N$, we obtain
\begin{equation}
\label{eqn:thmc}
|w_{j,k_i'}(i)| = |\gamma_j(i)||\alpha_{k_i'}(i)|,
\end{equation}
and we let $\alpha \equiv {\rm min}_{1 \le i \le n}|\alpha_{k_i}(i)|$. Clearly, $\alpha > 0;$ 
and (\ref{eqn:thmc}) allows us to assert that 
\[ 
     \forall \, j=1, \ldots, N \; {\rm and } \; \forall \, i=1, \ldots, n, \quad |w_{j,k_i'}(i)| \geq \alpha|\gamma_j(i)|.
\]

{\it iii.} We shall prove that 
$\{\gamma_{j_i}\,\alpha_{k_i'} : i = 1,\ldots, n\}$
is a basis for $\CC^n$, using the calculation from part {\it ii}.



Take any $i \in \{1, \ldots, n\}$.  Then, from $i-$dominance, there is $j_i$ such that
$\gamma_{j_i}(i) \neq 0$; and, hence, by taking the largest such value we can assert that
\[
    \exists\, j_i \;{\rm such}\,{\rm that}\; \forall \ell \ne j_i \quad  |\gamma_{j_i}(i)| \geq |\gamma_\ell(i)|.
\]
Consider the $n$ vectors, $w_{j_i,i} = \gamma_{j_i}\alpha_{k_i'} \in \CC^n, \, i=1,...,n,$
that is,
\[
     w_{j_i,i} = \gamma_{j_i}\alpha_{k_i'}
     = (\ldots, \gamma_{j_i}(h)\alpha_{k_i'}(h), \ldots) \in \CC^n, \;  h = 1, \ldots, n.
\]
Under the stronger assumption of strong $i-$dominance,
we can now verify that 
they form a basis for $\CC^n.$

{\it iv.}
Let
\begin{equation}
\label{eqn:thmh}
\sum_{i=1}^n\,b_i\gamma_{j_i}\alpha_{k_i'} = (0,\ldots, 0) \in \CC^n,
\end{equation}
where the notation for the sequence, $\{j_i \, : \,i=1,\ldots,n\},$ 
is that used to define strong $i-$dominance. In
this situation, whenever Equation (\ref{eqn:thmh})
is given, we shall show that each $b_i$ is $0.$ Thus, we can conclude that
$\{\gamma_{j_i}\alpha_{k_i'}\}$ is a basis for $\CC^n.$

Equation (\ref{eqn:thmh}) can be written as
$$
     \forall h = 1,\ldots, n, \quad \sum_{i=1}^n\,b_i\gamma_{j_i}(h)\alpha_{k_i'}(h) =0.
$$

Thus, for example, when $h=1,$ we have
$$
     b_1\gamma_{j_1}(1)\alpha_{k_1'}(1) + \sum_{i=2}^nb_i\gamma_{j_i}(1)
     \alpha_{k_i'}(1)=0.
$$     
If $b_1\ne 0,$ then 
$$
     \gamma_{j_1}(1)\alpha_{k_1'}(1) = -\sum_{i=2}^n\, \frac{b_i}{b_1}\,\gamma_{j_i}(1)\alpha_{k_i'}(1).
$$
By $1-$strong dominance, we have 
$|\gamma_{j_1}(1)| \,>\, |\gamma_{j_i}(1)|$ for each $i \in \{2,\ldots, n\}$; 
In particular, $j_1 \neq j_i$. 
By $1-$radiativity and the definition of $k_i'$, we have
$|\alpha_{k_1'}(1)| \ge |\alpha_{k_i'}(1)|$ for each $i \in \{2,\ldots, n\}.$ 
Consequently, we have the estimate,
$$
     \left|\sum_{i=2}^n\,\frac{b_i}{b_1}\, \gamma_{j_i}(1)\alpha_{k_i'}(1)\right| \le 
     \sum_{i=2}^n\,\left|\frac{b_i}{b_1}\right|\, |\gamma_{j_i}(1)| \, |\alpha_{k_i'}(1)| 
$$
$$
     < \sum_{i=2}^n\,\left|\frac{b_i}{b_1}\right|\, |\gamma_{j_1}(1)| \, |\alpha_{k_1'}(1)|\le
 |\gamma_{j_1}(1)| \, |\alpha_{k_1'}(1)|(n-1) \max_{2\le i\le n}\,\left|\frac{b_i}{b_1}\right|.
$$

We can now return directly to our task of 
proving that if Equation (\ref{eqn:thmh}) is given, then each $b_i$ is $0.$ If this is not
true, then let $b_h\ne 0,$ and assume $b_m$ is the largest such $b_h$ in the sense that 
$|b_h| \le |b_m|$ whenever $h \ne m.$
In particular, $b_m \neq 0,$ and (\ref{eqn:thmh}) implies $\sum_{i=1}^n b_i\gamma_{j_i}(m)\alpha_{k_i'}(m)=0.$ Therefore,
$$
    |b_m\gamma_{j_m}(m)\,\alpha_{k_m'}(m)| = 
    \left|\sum_{i \ne m}\,b_i \, \gamma_{j_i}(m)
\alpha_{k_i'}(m)\right| 
$$
$$
       \le \sum_{i\ne m}|b_i|\,|\gamma_{j_i}(m)\alpha_{k_i'}(m)|
       < \frac{|b_m|}{N - 1}\,\sum_{i\ne m}\,
       |\gamma_{j_m}(m)|\,|\alpha_{k_i'}(m)|,
$$
where the last inequality is due to strong $m-$dominance. Hence, we have
$$
     |\alpha_{k_m'}(m)|
     < \frac{1}{N - 1}\,\sum_{i\ne m}\,|\alpha_{k_i'}(m)|
     \leq \frac{1}{N - 1}\,\sum_{i\ne m}\,|\alpha_{k_m'}(m)|    
      = \frac{n-1}{N-1}\,|\alpha_{k_m'}(m)|,
$$ 
a contradiction. Here, we use the assumptions that $n \leq N$ and $N> 1.$ Thus,
$b_m=0$ and so $\{ \gamma_{j_i}\alpha_{k_i'}\}$ is a basis for $\CC^n.$  
\end{proof}


\begin{example}[A constructive component for Theorem \ref{thm:frame}]
\label{ex:constech}
In order to gain insight into the proof of Theorem \ref{thm:frame}, and
simultaneously to introduce a linear algebra approach for computational
reasons, let us proceed to prove Theorem \ref{thm:frame} by showing that
$\{w_{j, k_i'}\} \subseteq\CC^n$ spans $\CC^n$, and, 
hence, that it is a frame for  $\CC^n.$ The approach introduces ancillary 
$N \times n$ matrices, $(c(j,i)),$ that can be used for computation.



{\it i.} We begin, as in the proof of Theorem \ref{thm:frame} by using radiativity
to define $\{w_{j, k_i'}\}$, and 
note that ${\rm card}\{w_{j,k_i'}\} = nN \ge n$. Let $W = {\rm span}\{w_{j,k_i'}\} 
\subseteq \CC^n,$ and 
suppose $W \neq  \CC^n.$ Then, 
\[ 
   \exists \, y\in \CC^n\backslash \{0\} \; {\rm such \; that} \;\forall \, w \in W, \quad \langle y,w \rangle =0.
\]
Using the Hahn-Banach theorem our goal is to obtain a contradiction.

Notationally, let $w_{j,k_i'} \equiv w_{(j,i)} \in \CC^n.$ Hence, 
we can order the pairs $(j,i)$ from 1 to $nN.$ Instinctively, we would 
write $w_{(i,j)},$ as well as $c_{(i,j)}$ below, but the $i$ in these cases is dependent on
$k$, and this is our way of dealing with the lexicographic order $j,k$.

{\it ii.} Now, let 
$$
     y=\sum_{h=1}^n y(h)e_h \, , \quad e_h=(0,....0,1,0,...),
$$ 
where the "$1$" in the definition of 
$e_h$ appears in the $h$ position. It is {\it critical} to note that $e_h$ need not be in $W$ or the
orthogonal complement of $W$ in $\CC^n.$ 

Take any $w \in W$ so that
$$
w= 
\sum_{j=1}^N\,\sum_{i=1}^n\,c(j,i)w_{(j,i)} \in \CC^n,
$$ 
since each $w_{(j,i)} \in \CC^n.$ By hypothesis, we have
$$
   0=\langle w,y \rangle=\sum_{h=1}^n\overline{y(h)}\,\langle \sum_{j=1}^N\sum_{i=1}^n\,c(j,i)w_{(j,i)}\, , \, e_h \rangle,
$$ 
and so
$$ 
   0=\sum_{h=1}^n\overline{y(h)}\sum_{j=1}^N\sum_{i=1}^n \, c(j,i) \langle w_{(j,i)},e_h \rangle = 
   \sum_{h=1}^n\overline{y(h)}\sum_{j=1}^N\sum_{i=1}^n\,c(j,i)(\sum_{l=1}^ne_h(l)\overline{w_{(j,i)}(l)}) 
$$
$$
    = \sum_{h=1}^n\overline{y(h)}\sum_{j=1}^N\sum_{i=1}^n\,c(j,i)w_{(j,i)}(h),
$$
where
$$
   w_{(j,i)}(h)=w_{j,k_i'}(h)=\gamma_j(h)\alpha_{k_i'}(h).
$$
Combining these equations, we obtain
\begin{equation}
\label{eqn:thmd}
      0=\sum_{h=1}^n\overline{y(h)}\sum_{j=1}^N\gamma_j(h)\sum_{i=1}^nc(j,i)\alpha_{k_i'}(h),
\end{equation}
where $(y(1), \ldots, y(n))\neq(0, \ldots, 0)$ is fixed 
and $\{c(j,i) : j=1, \dots, N  \;{\rm and} \; i=1,\dots, n\}$ 
is any $nN$-tuple.

We can obtain the desired contradiction when we construct an $nN$-tuple $\{c(j,i)\}$ 
such that the right side of Equation (\ref{eqn:thmd}) is non-zero. With this contradiction 
we can then assert that $W=\CC^n$ and so $\{w_{(j,i)}\}$ is a frame for $\CC^n.$

{\it iii.} We write Equation (\ref{eqn:thmd}) as 
\begin{equation}
\label{eqn:thme}
    0=\sum_{i,j} c(j,i) \left( \sum_{h \in I} \overline{y(h)} \,{\gamma}_{j}(h)\,{\alpha}_{k_i'}(h)\right),  
\end{equation}
where $I = \{ h \in \{1,\ldots, n\} : y(h) \neq 0\}.$  Let $X = \{ (j,i) : j = 1,\ldots, N \; {\rm and}
\; i = 1,\ldots, n\}$ have the property that
$$
         \sum_{h \in I} \overline{y(h)} \,{\gamma}_{j}(h)\,{\alpha}_{k_i'}(h) \equiv d(j,i) \neq 0.
$$

{\it iv.}
If $X \neq  \varnothing,$ then set $c(j,i) = \overline{d(j,i)}$ for $(j,i) \in X$ and $c(j,i) =0$
otherwise. Then,
$$
0=\sum_{i,j} c(j,i) \left( \sum_{h \in I} \overline{y(h)} \,{\gamma}_{j}(h)\,{\alpha}_{k_i'}(h)\right) > 0, 
$$
and this contradicts Equation (\ref{eqn:thme}). Thus, $W = \CC^n.$

{\it v.}
Now assume $X = \varnothing,$ i.e., assume
$$
\forall\; i = 1,\ldots, n \; {\rm and} \; \forall\, j = 1,\dots, N, \quad
\sum_{h \in I}\overline{y(h)} \, \gamma_j(h)\, \alpha_{k_i'}(h) = 0.
$$
Then, for any $j_i \in \{1,\ldots, N\}$ corresponding to $i \in \{1,\ldots, n\},$ we have
\begin{equation}
\label{eqn:thmf}
     \forall \, i = 1,\ldots, n, \quad \sum_{h \in I}\overline{y(h)} \, \gamma_{j_i}(h)\, \alpha_{k_i'}(h) = 0.
\end{equation}
It is at this point that we 
go back to the proof of Theorem \ref{thm:frame}, using the
hypotheses of strong dominance and the properties of
$n$ and $N$, to choose $\{j_i : i = 1,\dots, n\}$ such that 
$$
       \{\gamma_{j_i}\,\alpha_{k_i'} : i = 1,\ldots, n\} \subseteq \CC^n
$$
is a basis for $\CC^n.$ Thus, we can conclude from Equation (\ref{eqn:thmf}) that 
$ y = 0 \in \CC^n,$ and this contradicts our assumption about $y.$ Therefore, $W = \CC^n.$

\end{example}

\begin{example}[An elementary example]
\label{ex:elemex}
Consider a sensing scenario $S$ with $K=1$ and with 
three sensors $s_1, \, s_2, \, s_3$, where $v_{1,1} =(3,1,1), \, v_{2,1} =(1,3,4)$,
and $v_{3,1} = (1,1,5)$. Let $H: \CC^3 \longrightarrow \CC^2$ be a health-space mapping
defined as the projection mapping onto the first two coordinates. Thus, we have
$H(v_{1,1}) = (3,1), \, H(v_{2,1}) = (1,3)$, and $H(v_{3,1}) = (1,1)$.
We define $\gamma_1(1) = 3, \,  \gamma_1(2) = 1, \, \gamma_2(1) = 1, \, \gamma_2(2) = 3, \,
\gamma_3(1) = 1$, and $ \gamma_3(2) = 1$. We also define $\alpha_1(1) 
= \alpha_1(2) = 1$. Then, we obtain
$H(v_{j,k})(i) = \gamma_j(i) \, \alpha_k(i)$, for $j = 1,\,2,\,3$ and $k=1$; 
and so $S$ is a separable sensing scenario. Further, $S$ satisfies the criterion
for radiativity, since $\alpha _1(i) \neq 0$ for $i = 1,\,2$.
Finally, we can see that $S$ is strongly $i$-dominant for $i = 1,\,2$ by defining
$j_1 = 1,  j_2 = 2$. We now apply Theorem \ref{thm:frame}, and can assert that
$\{w_{j,k} : j = 1,\,2,\,3 \, {\rm and}\, k=1  \} = \{ (3,1), \,  (1,3 ), \,  (1,1) \}$
is a multiplicative frame for $\CC^2$.
\end{example}






\section{Data base and DFT turbine simulation}
\label{sec:datasim}

\subsection{Detection strategies}
\label{sec:detection}

Determining the dimension $n$ of health space is critical to obtaining useful algorithms based on 
this theory. To this end, we can
divide problems into two classes depending on whether $n$ is known a priori or a posteriori. 
Both cases can arise naturally 
and the theory will be applicable in both cases, but the algorithms for implementation may be different.

{\it a.} In the first case, $n$ is known in advance. This will arise in problems where a detection 
strategy has already been developed. In fact, for a number
of engineering problems, detectors already exist for single sensor output. SONAR, RADAR, and machine 
diagnostic problems all have detectors associated with a single installation or sensor. In the DFT/turbine 
example of this section,
we assume that the existence of engine
faults can be determined by looking for certain spectral lines.
In problems of this sort, given accurate sensor data from a working sensor, the
detector can determine the status of a turbine. This detector maps the output of a 
single sensor to health space and thus
determines $n$. There is dimension reduction when the data from multiple sensors is 
combined into 
something the
detector can process and also when the detector produces output of lower dimension than 
the original data stream.
Subsections \ref{sec:data} - \ref{sec:snr} deal with this case.

{\it b.} The second case arises when a detection strategy is not known in advance. Consider a 
large collection of data taken by 
several sensors. For example, each sensor may be recording images or multispectral data, 
that must be compressed 
using a dimension reduction technique before being transmitted to a processing center. 
It is possible that no detector is 
available for the compressed 
data stream.

In this case, a machine learning algorithm may be used to determine the 
state of the situation observed by 
the sampled data, e.g., see \cite{geor2012}
and the remote sensing applications analyzed in 
\cite{CzaManMcLMur2016},\cite{BenCzaDobDosDuk2016},
\cite{CloCzaDos2017}. For example, given a set of known positive 
and negative detection images, 
an optimal Bayesian detector 
such as a matched filter or bank of matched filters may be constructed for the compressed data. 
The output of these filters 
would then determine  health space and the dimension $n$. 
We also note 
that matched filters are only one example, but 
other machine learning techniques may also be used. 


\subsection{Data base construction}
\label{sec:data}

\begin{figure}[htbp]
\centering
\includegraphics[height=0.2\paperheight]{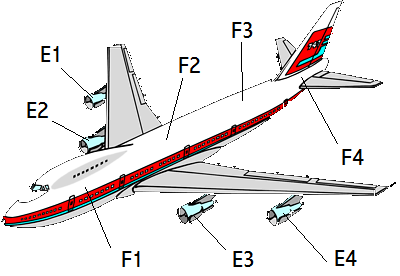}
\caption{Airplane with engine sensors $E_j$ and fuselage sensors $F_j$}
\label{fig:plane}
\end{figure}
The theorems and corollaries of Section \ref{sec:rst} provide a mathematical framework 
for the analysis of sensing problems. We now give a detailed example of how this framework 
applies to a practical simulation for  analyzing the DFTs of sensed vibration data. This is, in a 
way, our canonical situation in that we developed much of the theory of reactive sensing by analyzing
the DFT case.
 
To this end, we now construct a data base to analyze basic vibrations associated with an airplane.
We use public
domain specification data of major jet engine manufacturers, viz., Boeing, GE, and Pratt-Whitney,
see \cite{geav}, \cite{henc2012}, \cite{yoon2001}.
The engine noise models are  simplified since our immediate purpose here is to provide an illustrative
sensing scenario, not to solve a complex avionics problem.

{\bf Airplane and vibration sensors.} In Figure \ref{fig:plane}, we have an airplane with
four engines, $E_j, \, j=1,2,3,4$, and on
each engine there is a vibration sensor
$s_j$.  We also designate sections, 
$F_j, \, j=1,2,3,4$, of the fuselage, that may also
have sensors attached. 

{\bf Turbine and gear assumptions.} We begin by assuming each engine consists of two turbines spinning
at different speeds.
Each turbine has fan blades attached to it as well as complex gear assemblies. For simplicity
we assume there are only 3 gears in each assembly, one connected to the first turbine and two
connected to the second. We assume the basic sounds generated by the engines are spectral
lines associated with each of their components. Specifically, each turbine will generate one
frequency component at the shaft rotation frequency and another for the blade frequency, which
is assumed to be equal to the shaft frequency multiplied by the number of blades. Each of the 3
gears is also assumed to generate a frequency component equal to the gear ratio multiplied by its
associated turbine frequency. This gives a total of 7 spectral lines associated with each engine.
We assume these spectral lines are unique for each engine. This is a 
somewhat unrealistic constraint, but it provides
an elementary setting in order to evaluate
our theory.

{\bf Vibration sensor output.} So far, we have described the modelling of the engine vibrations 
themselves. We must now determine what
the output of the vibration sensors will be. We assume that the sensor on each engine
is primarily responsible for reporting data for that engine, see Subsection \ref{sec:sepsensscen}
for the notion of primary responsibility. 
We also assume that each sensor detects
vibrations from the other engines but at a reduced volume. This last assumption is quite natural, but it is
significant since it forms the basis for the multiplicative structure developed mathematically in
Section \ref{sec:frames}. Finally, we assume there are noise sources, e.g., air moving across the
fuselages, that are also detected by the sensors.

{\bf Mixing matrix.} To generate the output of the sensors, we define a mixing matrix that 
determines what proportion of each
engine's vibrations are sensed at each sensor. Specifically, we define a 
mixing matrix $A = (a_{j,h})$  where each 
entry $a_{j,h}$ gives the relative volume of engine $E_h$ reported by sensor $s_j$. For example, we might
have $a_{1,1}=1,\ a_{1,2}=0.1,\ $ and $a_{1,3}= 0.01$, which would mean sensor $s_1$ would register 
engine $E_1$ at nominal volume 1, engine $E_2$ at 10 dB down from that level, and engine $E_3$ a further
10 dB down (20 dB down in all), see \cite{henc2012}.

{\bf The data base.}
We choose the engine vibration outputs we wish to
combine and apply the mixing matrix to obtain the desired sensor output. For example, 
to manufacture  data for 
a properly functioning airplane, we take data that models the noise of each properly working engine, and combine
these data using the mixing matrix. To 
manufacture data for a gear fault in $E_1$, replace the properly working engine
noise data for $E_1$ with the gear fault noise and reapply the mixing matrix. In this fashion we can 
manufacture
data for both properly working and faulty engine operation, as sensed by each of the sensors.

The sensor outputs can now be used to generate the data base. 
The output of each sensor is processed by an 8192 point DFT in blocks
giving 4 copies of ${\mathbb R}^{8192}$.
We wish to map this data from $({\mathbb R}^{8192})^4$ into a space
where the state of the engines can be assessed using an appropriate
detector. It is at this stage that we invoke 
dimension reduction technology, albeit in an elementary way.
In fact, we
assume that the health of each engine can be estimated
based on the magnitude of the 28 relevant spectral lines (7 for each
engine). Thus, health space is ${\mathbb R}^{28}$. We implement a simple
detector,
that operates on ${\mathbb R}^{28}$ and is capable of distinguishing normal
operation, basic engine faults, and catastrophic engine failure. For example, basic faults, such as a 
missing gear tooth, are registered in terms of large spectral values; while 
a catastrophic engine failure would be sensed as missing all 
spectral data from that engine. For more on detectors, and detection and classification
strategies, see \cite{vant2002}.


\subsection{The data base as a sensing scenario}
\label{sec:datasensscen}

From a mathematical point of view, what we constructed 
in Subsection \ref{sec:data} is a sensing scenario,
as defined in Subsection \ref{sec:sepsensscen}. In fact, the sensors $s_j$
attached to each engine  produce spectral lines by repeatedly computing DFTs on the vibration data. Thus,
at fixed discrete times $k$, we obtain vectors $v_{j,k} \in \CC^M$
containing frequency information for each sensor $s_j, j = 1, \ldots, N$. Here, $M = 8192$ and 
$\{1, \ldots, M\}$ is the frequency domain of the DFT.    
Using the data base, we compute the vibration noise coming from each engine, and 
apply a mixing matrix to compute the sensor outputs. The initial engine vibrations produce frequencies, 
$f \in \{1, \ldots, M\}$, 
with varying amplitudes, ${\widehat \alpha}_k(f)$. These values are generated by 
the engines, independent of the sensors, and  hence the ${\widehat \alpha}_k$ are independent
of $j$. The mixing matrix determines the relative weights that are used to combine the engine vibrations, thus
determining the values of ${\widehat \gamma}_j(f)$ for each spectral line $f$. This mixing is independent of
time, and thus does not depend on $k$. Therefore, for each $s_j$, we have 
\[
     v_{j,k}(f) ={\widehat \gamma}_j(f){\widehat \alpha}_k(f), \quad f \in \{1, \ldots, M\}.
\]
This means that our sensing scenario is {\it pre-separable} as defined in Definition \ref{defn:sep}.

Further, we define an elementary health space mapping,
\[
       H: {\mathbb R}^{8192} \longrightarrow {\mathbb R}^{28},
\]
by selecting only the coordinates of the 28 relevant spectral components described above.
${\mathbb R}^{28}$ will be the associated health space, and this $H$ will be the 
health space mapping. It is a projection mapping, and is equivalent to PCA 
under fairly benign assumptions.

We note that $H$ meets the criteria of Proposition \ref{prop:sep2}. Specifically, $H$ is 
the identity mapping on the subspace ${\mathbb R}^{28}$ of relevant spectral frequencies, and it is
the 0-mapping everywhere else. Since we have shown our scenario is pre-separable,
Proposition \ref{prop:sep2} allows us to conclude that we
have a {\it separable} sensing scenario with health space mapping $H$.

Using the definitions in Subsection \ref{sec:bmfm}, we can now define basis and frame mappings for 
the scenario. From a practical point of view, this is 
a necessary task, since the sensors are each generating 
a copy of ${\mathbb R}^{8192}$, whereas the health space mapping $H$ only operates on one copy. 
In any case, we require that the basis mapping $B$ and the frame mapping $F$ both map
$({\mathbb R}^{8192})^4$ to
health space:
\[
      B,F: ({\mathbb R}^{8192})^4 \longrightarrow {\mathbb R}^{28}.
\]
The basis mapping $B$ will map the magnitude of each of the 7 spectral lines
for $E_j$ recorded by $s_j$ to one $7$-dimensional subspace of ${\mathbb R}^{28}$. $B$ will
ignore the rest of the reported data. We shall define the frame mapping $F$ as the 
sum of the magnitudes of each of the spectral lines as
recorded by all the sensors $s_j$. Thus, $F$ is the magnitude sum frame mapping defined 
in Definition \ref{defn:BFmappings}.

\begin{rem}[Additional sensors]
We note here that there may be additional sensors attached to the
airplane, which, for example, might monitor fuselage vibrations. These have not been included in the current
data base and analysis, but one can imagine applying reactive sensing theory to these sensors as well.
In fact, one envisages tuning such sensors 
to sense the shape of the broadband noise envelope, and not only processing the sharp spectral lines
associated with the engines.
\end{rem}



\subsection{DFT - turbine simulation}
\label{sec:sim}

We now use the data base described in Subsections \ref{sec:data} and
\ref{sec:datasensscen} to illustrate a practical sensing scenario. We
start with the 4 sensors $s_j$, which will each hear clearly its engine
$E_j$, and we assume it will also hear the other engines
$E_{\ell}, {\ell}\neq j$, at a lower volume, 10dB or more down
along with a significant level of noise in the system.
The sensors produce spectral output by computing a DFT at each time $k$.
Here we use an 8192 point DFT. As described in Subsection
\ref{sec:datasensscen}, the sensed outputs of the spectral
components from a given sensor at a given time ($v_{j,k}$)
is the product of the volume of the components produced by the
engines ($\alpha_k$) with the sensor's ability to hear the
components ($\gamma_j$).

The assumption that a given sensor can hear all of the engines
amounts to describing the covering $ T_j = S$ for each $j=1, \ldots, 4$,
where S is the frequency spectrum produced by all of the engines.
Here, ${\rm card}\,(S) = M = 8192$.
Note that a sensor $s_j$ can be thought of as bearing primary
responsibility for the spectral lines produced by
the engine $E_j$ to which it is attached
(once again see Subsection \ref{sec:sepsensscen}
about primary responsibility).Further, since every sensor can hear
every engine, albeit at a possibly decreased volume, we can assert
that the scenario is harmonious as in
Definition \ref{defn:disjharm}.

The magnitudes of the spectral lines give an element of
${\mathbb R}^{8192}$ for each sensor. The basis and frame
mappings both combine the sensor output to produce a single element
of health space, ${\mathbb R}^{28}$, which
can then be sent to the detector.
The detector processes the magnitudes of the 28 significant spectral
lines to determine the fault condition of the airplane. We note that
this is consistent with the detection strategy described in
Subsection \ref{sec:detection}, part {\it a}. Specifically, we are
assuming that the 28 relevant spectral lines are known, along with the
values for the magnitudes of these lines in a working airplane. We can
therefore apply the detector to the output of the
basis and frame mappings to obtain detection and classification results.

\begin{figure}[htbp]
\centering
\includegraphics[height=0.2\paperheight]{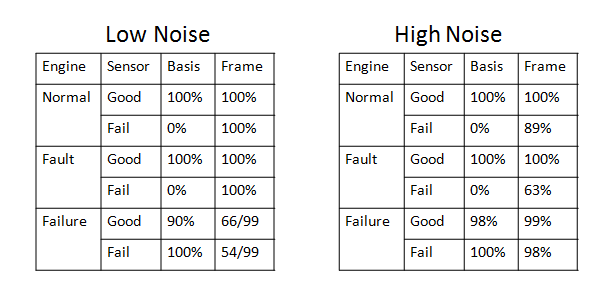}
\caption{Results}
\label{fig:results}
\end{figure}


\subsection{Results}
\label{sec:results}


We now have everything in place to apply the theorems and corollaries of Subsection \ref{sec:bt}.
Thus, we can look at the numerical data in Figure \ref{fig:results} to see the relation between
these simulated results and the theory. These results were generated by taking samples of data from the
database generated in Subsection \ref{sec:datasensscen}. Each sample was processed with an 8192 
point DFT. A total of 1024 samples for each of the three different engine states was chosen at random.
The three states were the following: 
normal operation of engine $E_1$, gear fault in engine $E_1$, and complete failure of engine $E_1$. 

Figure \ref{fig:results}  shows the results of applying the process described in 
Subsection \ref{sec:sim} to the data base generated in 
Subsections \ref{sec:data} and \ref{sec:datasensscen}. We consider data from the data base
for 12 different conditions: the 3 different engine conditions for $E_1$ described above
(normal operation, gear fault, and complete engine failure) and 2 different sensor conditions 
for $s_1$ ($s_1$ in good operating condition, and $s_1$ failure) under the 2 conditions of low background
noise and high background noise. For each of these combinations of conditions we report the percentage 
of {\it correct detections}. For the case of complete engine failure in the low background noise environment
we report 2 numbers for the frame detector: the first (66 and 54) are strictly engine failure reports, while the 
second (99 and 99) are
combined engine failure and engine fault reports. 


The inclusion of additional sensor results in the low noise 
background tests has caused the frame based approach to confuse fault conditions with failure conditions. Often, 
multiple fault conditions would be reported. 

{\bf Good (operational) sensor data from Figure \ref{fig:results}.}
Theorems \ref{thm:basismapping} and \ref{thm:framemapping} guarantee that the image
of certain large sets will form a basis or a frame, respectively, in health space. In either case, 
we have the ability to span
${\mathbb R}^{28}$ and, therefore, our detector should do a good job of classifying the engine
condition. Figure \ref{fig:results} shows that when all sensors are operational (Good), normal operation of the 
engine is correctly identified by both the basis and the frame 
mappings in both low and high noise environments. This is reflected by the $100\%$ 
listed in all 4 possible places of the first row of Figure \ref{fig:results}.
Similarly, both basis and frame processing 
for operational (Good) sensors correctly identify fault conditions. This is
reflected by the $100\%$ 
listed in all 4 possible places of the third row of Figure \ref{fig:results}.
The situation is less uniform at identifying engine failure when all sensors are operational (Good).
This is the content of the percentages listed in the 4 possible places of the fifth row of Figure \ref{fig:results}.
When the ambient noise in the system is low, basis processing had a $90\%$
success rate at identifying engine failure. On the other hand, in the same low noise case, 
the frame mapping sometimes 
reported an engine fault, or multiple engine faults, when the engine had actually failed. 
This is reflected by the $66\%$ that is listed. The reason for this is not surprising
given the additional noise contributed by the combination of sensors, and
the sensitivity of frames in integrating noise into its analysis. 
The $99\%$ that is listed is the percentage of detecting a fault or engine failure.
In the high noise environment, basis and frame processing
have an excellent success rate, $98\%$ and $99\%$, respectively,
at identifying engine failure when all sensors are operational (Good).

{\bf Failed sensor data from Figure \ref{fig:results}.} The situation is different when a sensor fails. 
Corollary \ref{cor:basismapping} implies that in the case of a 
basis mapping, there will not necessarily be a corresponding basis in health space; whereas 
Corollary \ref{cor:framemapping} implies that in the case of a frame mapping, there will be
a corresponding frame in health space. The 
results in Figure \ref{fig:results} show that when a sensor fails (Fail), 
 the basis mapping 
always indicates an engine failure. This is reflected by the two $0\%$ listings in the second
row of Figure \ref{fig:results}. Of course, sensor failure does not imply engine failure.
On the other hand, when a sensor fails,
the frame mapping still correctly 
identifies a normal working engine. This is reflected by the $100\%$  and $89\%$ listings in
the second row of Figure \ref{fig:results}.
Further, when a sensor fails, the frame mapping can still distinguish a
normal working engine from a fault or an engine failure, while the basis mapping cannot.
This is seen in the fourth line of Figure \ref{fig:results}, where the basis mapping reports $0\%$ 
fault detection while the frame mapping reports $100\%$ in the low noise case and $63\%$
in the high noise case.  Of course, the basis mapping will report an engine failure in all cases when
the sensor fails, and in the case of an actual engine failure, it will be correct. This is seen in the two 
$100\%$ entries in the last row. As noted above, the frame mapping has some trouble distinguishing 
between a fault and an engine failure in the low noise case (the $54\%$ in the bottom row), but does
recognize that there is a problem (combined fault and engine failure reports given by the $99\%$ in the 
bottom row). In the high noise case the frame mapping reports engine failure 
correctly $98\%$ of the time. 



Thus, when a sensor
fails, the frame mapping can still give some data about the engine. The
basis mapping, however, cannot distinguish between a sensor failure and catastrophic
engine failure. 
In high noise environments, we note that the performance of the frame
mapping degrades due to noise, but it is still useful, while the basis mapping is not. 


\begin{figure}[htbp]
\centering
\includegraphics
{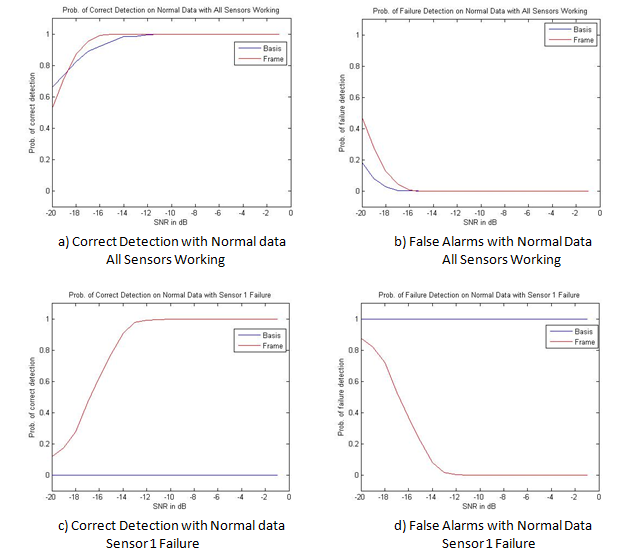}
\caption{Detection and false alarms}
\label{fig:det_fa}
\end{figure}


\subsection{SNR considerations}
\label{sec:snr}


The analysis in Subsection \ref{sec:results} illustrates some of the features of the frame and basis approaches 
at a fixed SNR. If we vary the SNR 
we can gain some insights into the different specific behaviors of each approach. 
Figure \ref{fig:det_fa} shows what 
happens to the detector's ability to determine if there has been a catastrophic 
failure as the SNR varies from -20 dB to 
0 dB.

In each of the four cases in Figure
\ref{fig:det_fa} we have assumed all of the engines are working normally. 
Since the input 
data is for a properly working engine, 
all failure and fault detections are false alarms as the result of noise. We 
only consider statistics 
for normal engine operation and 
catastrophic failure. We note that fault detections will count as neither. Further, 
since all failure and fault detections are the result of noise,
we can quantify the performance of each approach as a function 
of SNR, that is, we can plot the probability of correct detection and the probability of false alarm as 
functions of SNR in both the cases for properly working sensors and for sensor $s_1$ failure. 
We note that for these results, we have changed the mixing matrix,  increasing the levels at 
which sensor $s_j$ will hear engine $E_{\ell}$, 
for ${\ell} \neq  j$, from 10dB down to 5dB. As before we are only looking at 
results for engine $E_1$. 
Figure \ref{fig:det_fa}a shows the correct detections when all sensors are working. 
When the SNR is high, both the frame 
and basis approaches give perfect detection, and, as the SNR declines, both show a 
decline in correct detections. 
When the 
SNR falls below -18dB we see that the basis approach has a higher probability 
than the frame approach of 
detecting that the engine is performing 
correctly. This is expected since by looking at other sensors, the frame approach 
will be forced to contend with additional noise.  
It is interesting and somewhat counter-intuitive that the frame approach has a higher probability 
of correct detection 
between -18 and -12 dB SNR. It turns out that as the noise increases, the basis approach 
has a greater probability of detecting a fault due to a short term noise spike while the frame approach will average out large noise spikes. However, at low SNR the basis approach detects fewer failures than the frame approach. 
This can be seen in Figure \ref{fig:det_fa}b which shows the false alarms for both approaches. 
Here it is clear that the elevated noise levels cause the frame approach to report more false alarms.  

When sensor $s_1$ fails, however, the graphs look significantly different. Figures \ref{fig:det_fa}c 
and  \ref{fig:det_fa}d show the correct detections and false alarms when sensor $s_1$ 
has failed. Again the data studied is for normal operation, that is, a properly working engine. 
The frame approach performs well when the SNR is high, giving perfect detection with no false alarms. 
As the SNR declines we see that for the frame approach, the correct detections decline and 
the false alarms increase somewhat 
faster than when sensor $s_1$ is working. This is not surprising since effective signal strength 
is substantially lower without $s_1$. The basis approach, however, fails completely as it constantly 
detects a catastrophic failure regardless of the input data or SNR. This is where the ability of the frame 
approach to utilize the overdetermined nature of frames comes into play and gives a substantial 
performance improvement.


\section{Epilogue and future}
\label{sec:epilogue}

{\it Reactive sensing} is an evaluative process to determine the behavior
of an object when primary sensor sources of information about the object become
unavailable, so that any information must be obtained from the intelligent use
of available secondary sensor sources. With regard to a host of applications, e.g.,
Section \ref{sec:ex}, loosely connected in terms of machine health, our approach
was in terms of physical and mathematical modeling, which in fact interleave 
and are equivalent.

Our initial, main observation with regard to {\it physical modeling}
was noting the centrality of 
combining and relating the {\it volume} of  a scene and the {\it sensitivity} of sensors
evaluating that volume, e.g., see Definition \ref{defn:sep} and Example \ref{ex:dft}
This, in turn, led to the idea of factoring sensor data, and formulating useful
notions for this physical modeling, viz., radiative, dominant, and harmonious
sensing scenarios, defined in Definitions \ref{defn:raddom} and \ref{defn:disjharm} and 
analyzed throughout Section \ref{sec:mathmodel}.
This factoring of sensor data, coupled with the need to implement our
physical modeling led to the development of
multiplicative frames, that is the centerpiece of our {\it mathematical modeling}.
The theory for this development required the content of Section \ref{sec:frames}.

Reactive sensing, integrating our physical and mathematical modeling,
fits naturally into the context of spatial {\it super-resolution}. In fact, as noted in Subsection \ref{sec:back},
the secondary 
sensor sources can be 
considered analogous to
the role of obtaining a high resolution (HR) image from observed multiple low-resolution (LR) images. 
In classical super-resolution, the multiple LR images represent different ``snapshots" of the same scene, 
and can be combined
to give the desire de facto HR image. In our case, the LR images correspond to secondary
sensors and the HR image corresponds to the primary object, that could be disabled or 
whose primary sensor
is not functioning.

Our theory of reactive sensing is given in Section \ref{sec:rst}. First, we establish the necessity of 
frames, as opposed to bases, in order to model, both mathematically and effectively, 
the role of secondary sensor sources, see Subsection
\ref{sec:bmfm}. Then, in Subsection \ref{sec:bt}, we prove the fundamental theorems, that are
the basis of our proposed implementation. 

Health space $\CC^n$ (Definition \ref{defn:sep}) and dimension
reduction are part and parcel of the same idea in reactive sensing with regard to this implementation.
In Subsections \ref{sec:results} and \ref{sec:snr}, the theory of Section \ref{sec:rst} 
is applied to a DFT turbine simulation.

The experimental results in Section \ref{sec:datasim} show the efficacy,
with some natural qualifications, of frame based reactive sensing 
theory in detecting engine conditions correctly even in the event of sensor failure. 
The following are our immediate reactive sensing tasks.

\begin{itemize}

\item 
Apply reactive sensing theory to a
spectrum of real data sets, see, e.g., Section \ref{sec:ex}, and 
formulate realistic, useful reactive sensing implementations
for such sets. 
This will require more sophisticated dimension reduction technology
than PCA, see, e.g., \cite{BenCza2015}.

\item Evaluate more deeply the effect of SNR on reactive sensing for
scenarios with and without sensor failure. As our simulations have shown, this is a
challenging task.

\item Develop multiplicative frames as a full-fledged mathematical theory;
and especially analyze its relation with frame multiplication theory
and group frames, see \cite{AndBenDon2019}.

\end{itemize}

\section{Acknowledgements}

The authors gratefully acknowledge the support of MURI-ARO Grant W911NF-09-1-0383. The
first named author also gratefully acknowledges the support of 
DTRA Grant HDTRA1-13-1-0015, and 
ARO Grants W911NF-15-1-0112,  
W911NF-16-1-0008, and W911NF-17-1-0014. This litany of good fortune involved other projects,
but also reflects the elephantine gestation period it took for the authors' original
sensor factorization idea to crystalize.
Both authors have also benefitted from insightful observations,
technical discussions, interest, and support by the 
following:
Drs. David Colella and David Clites (The MITRE Corporation);
Kevin Hencke (Systems Engineering Group, Inc.
of the Telephonics Corporation, and formerly of the Norbert Wiener Center (NWC));
Tej Phool (MimoCloud);
Professors Richard Baraniuk (ECE, Rice
University), 
Alexander Barg (ECE, University of Maryland), Rama Chellapa (ECE, University of Maryland), 
and Wojciech
Czaja (Mathematics, University of Maryland);
Drs. Alfredo Nava-Tudela, Matthew
Begu{\'e} (formerly of the NWC), and Celia Rees Evans. Drs. Nava-Tudela and Evans
are Scientific Development Officers of the NWC. Finally, we are especially
appreciative of the continued confidence of Drs. Liyi Dai, Joseph Myers, and 
Jay Wilkins Jr.


\bibliographystyle{abbrv}
\bibliography{2018-10-06JBbib}

\end{document}